\documentclass{article}
\usepackage{blindtext}
\usepackage[margin=1.2in]{geometry}
\usepackage{xcolor}
\usepackage{color,soul}
\usepackage{hyperref}
\usepackage[utf8]{inputenc}
\usepackage{amsmath}
\usepackage{amsthm}
\usepackage{amsfonts,mathrsfs}
\usepackage{amssymb}
\usepackage{graphicx}
\usepackage{enumitem}
\usepackage{color}
\usepackage{verbatim}
\usepackage[normalem]{ulem}
\theoremstyle{plain}
\newtheorem{theorem}{Theorem}[section]
\newtheorem{lemma}[theorem]{Lemma}

\newtheorem{corollary}[theorem]{Corollary}

\theoremstyle{definition}

\newcommand{\F}{\mathbb F}
\newcommand{\N}{\mathrm N}

\newcommand{\cC}{\mathcal C}

\newcommand{\Aut}{\mathrm{Aut}}

\newcommand{\GL}{\mathrm{GL}}

\newcommand{\PG}{\mathrm{PG}}
\newcommand{\GaL}{\Gamma\mathrm{L}}
\newcommand{\PGaL}{\mathrm{P}\Gamma\mathrm{L}}

\newcommand{\RN}[1]{%
  \textup{\uppercase\expandafter{\romannumeral#1}}%
}
\newcommand{\rn}[1]{%
  \textup{\lowercase\expandafter{\romannumeral#1}}%
}

\title{On the equivalence issue of a class of $2$-dimensional linear maximum rank-metric codes}
\author{Somi Gupta, Giovanni Longobardi and Rocco Trombetti}
\date{}

\begin{document}

\maketitle

\abstract{In [A. Neri, P. Santonastaso, F. Zullo.
\textit{Extending two families of maximum rank distance codes}], the authors extended the family of $2$-dimensional $\F_{q^{2t}}$-linear MRD codes recently found in [G.~Longobardi, G.~Marino, R.~Trombetti, Y.~Zhou.
\textit{A large family of maximum scattered linear sets of $\PG(1,q^n)$ and their associated MRD codes}]. Also, for $t \geq 5$ they determined  equivalence classes of the elements in this new family and provided the exact number of inequivalent codes in it. In this article, we complete the study of the equivalence issue removing the restriction $t \geq 5$.  Moreover, we prove that in the case when $t=4$, the linear sets of the projective line $\mathrm{PG}(1,q^{8})$ ensuing from codes in the relevant family, are not equivalent to any one known so far.}

\bigskip


\medskip


\section{Introduction}
Let $n \in \mathbb{N}$, $q$ a prime power and denote by ${\mathcal L}_{n,q}[x]$ the vector space of all {$\mathbb{F}_q$-linearized polynomials} with coefficients over $\F_{q^n}$ (also known as $q$-{\it polynomials}); i.e., $${\mathcal L}_{n,q}[x]= \Big  \{\sum_{i=0}^{r}c_ix^{q^i} \, \colon \, c_i \in \F_{q^n} \,\text{ and }\, r \geq 0 \Big \}.$$ It is well known that \[\tilde{{\mathcal L}}_{n,q}[x] := {\mathcal L}_{n,q}[x]/(x^{q^n}-x) = \Big  \{\sum_{i=0}^{n-1}c_ix^{q^i} \, \colon \, c_i \in \F_{q^n} \Big \},\] is isomorphic to the algebra ${\rm End}_{\F_q}(\F_{q^n})$ of the $\F_q$-linear endomorphisms of $\F_{q^n}$, \cite{lidl_finite_1997}. Moreover, the distance between two different such endomorphisms $f(x)$ and $g(x)$, defined as 
$$d(f,g)={\rm dim}({\rm im}(f-g)),$$ gives rise to a metric on $\tilde{{\mathcal L}}_{n,q}[x]$, which is called the {\it rank metric}. In general, if $\mathcal{C}$ and $\mathcal{C}'$ are two subsets of $\tilde{{\mathcal L}}_{n,q}[x]$, we say that they are {\it equivalent} if there exist $\rho \in {\rm Aut}(\F_{q^n})$ and two permutation polynomials $L_1(x)$ and  $L_2(x)$ such that \[ \mathcal{C}'= \{ (L_1 \circ f^{\rho} \circ L_2 )(x) \, : \, f \in \mathcal{C} \},\] where $f^{\rho}(x)=\big ( \sum a_ix^{q^i} \big )^{\rho} = \sum a^{\rho}_ix^{q^i}$. We denote by $\tau=(L_1,L_2,\rho),$ the equivalence defined above and consequently denote by $\tau(f)$ the elements in $\mathcal{C}'$, see \cite{morrison_equivalence_2013}. 

A \textit{linear rank-metric} code ${\cal C} \leq \tilde{{\mathcal L}}_{n,q}[x]$ is an $\F_q$-subspace of the vector space $\tilde{{\mathcal L}}_{n,q}[x]$ endowed with the above mentioned rank metric $d(\cdot,\cdot)$. 
  The minimum of the set of all distances between two distinct words of a linear rank-metric code ${\cal C} \leq \tilde{{\mathcal L}}_{n,q}[x]$, $| \cC| \geq 2$, is denoted by $d_{\cal C},$ and one has \[d_{\cal C}={\rm min}\{{\rm dim}({\rm im}(f)) \,:\, f \in {\cal C} \setminus \{0\}\}.\] 
 
 If ${\rm dim}_{\F_{q}}({\cal C})=k$, then $d_{\cal C}\leq n-\frac{k}{n}+1$ and in the case where the equality is attained, ${\cal C}$ is said to be a {\it linear maximum rank-metric code} (linear MRD code, for short). 
 
 \medskip
 
\noindent Due to their application in network coding, the interest
in the rank-metric codes has intensified over the past years and many recent papers have been
devoted to their study. In addition, linear rank metric codes are also linked with interesting algebraic and geometric objects like {\it $q$-polymatroids} and {\it linear sets} of the projective space, see \cite{qpoliymatroids, connection}. 
 
 \medskip
\noindent Motivated by what is stated in Section $5$ of \cite{neri_santonastaso_zullo}, in this article, we will focus on linear MRD codes with minimum distance $d=n-1$ that are $\F_{q^n}$-subspaces of $\mathcal{\tilde{L}}_{n,q}[x]$. Precisely, in \cite[Theorem 3.8]{neri_santonastaso_zullo}, the authors introduce a new class of such MRD codes extending two constructions exhibited in \cite{longobardi_zanella} and \cite{longobardi_marino_trombetti_zhou}. This family is denoted by the symbol ${\mathcal C}_{h,t,s}$, where $t \in \mathbb{N}$ and $t \geq 3$, $h \in \F_{q^{2t}}$ such that $h^{q^t+1} = -1$ and $s \in \mathbb{Z}$ such that $(s,2t)=1$. Also, in \cite{neri_santonastaso_zullo}, it was proven that if $t >4,$ then ${\mathcal C}_{h,t,s}$ is not equivalent to any other example known in the literature so far.

Moreover, the following result was proven concerning the equivalence of two elements in ${\mathcal C}_{h,t,s}$.

\begin{theorem}\cite[Theorem 4.6]{neri_santonastaso_zullo}\label{th:equivalencebetween_NSZ} Let $t \geq 5$ and consider \,${\mathcal C}_{h,t,s}$ and ${\mathcal C}_{k,t,\ell}$ such that $(s,n)=1=(\ell,n)$. Then the codes ${\mathcal C}_{h,t,s}$ and  $\mathcal{C}_{k,t,\ell}$  are equivalent if and only if one of the following collection of conditions are satisfied:
\begin{itemize}
\item[I.\,] $s\equiv\ell \pmod n$, and there exists $\rho \in \Aut(\F_{q^n})$ such that \[\rho(h)= \begin{cases} \pm k, & \, \text{\rm if} \, t \not\equiv 2 \,\pmod 4 \\ \lambda k, \text{ \rm where }  \lambda^{q^2+1} = 1 & \text{ \rm if} \, t \equiv 2 \,\pmod 4\,  \end{cases} \]\\

\item[II.\,] $s\equiv-\ell \,\pmod n$, and there exists $\rho \in \Aut(\F_{q^n})$ such that \[\rho(h)= \begin{cases} \pm k^{-1}, & \, \text{\rm if} \, t \not\equiv 2\pmod \,4 \\ \lambda k^{-1}, \text{ \rm where }  \lambda^{q^2+1} = 1 & \text{ \rm if} \, t \equiv 2 \pmod \,4  \end{cases} \]\\
\item[III.\,] if $s\equiv(t-1)\ell \, \pmod n$, $t$ even, and there exists $\rho \in \Aut(\F_{q^n})$ such that \[\rho(h)= \begin{cases} \pm k, & \, \text{\rm if} \, t \not\equiv 2 \pmod 4 \\ \lambda k, \text{ \rm where }  \lambda^{q^2+1} = 1 & \text{ \rm if} \, t \equiv 2 \pmod 4 \end{cases} \]\\
\item[IV.\,] $s\equiv(t+1)\ell \, \pmod n$, $t$ even,  and there exists $\rho \in \Aut(\F_{q^n})$ such that \[\rho(h)= \begin{cases} \pm k^{-1}, & \, \text{\rm if} \, t \not\equiv 2 \, \pmod 4 \\ \lambda k^{-1}, \text{ \rm where }  \lambda^{q^2+1} = 1 & \text{ \rm if} \, t \equiv 2 \,\pmod 4 \end{cases} .\]\\
\end{itemize}
\end{theorem}
Clearly the necessity in the statement of this theorem, still holds true if $t=3$ and $4$. However, as underlined in the initial part of \cite[Section $4$]{neri_santonastaso_zullo}, in such cases its sufficiency part, and more generally the equivalence issue for the codes $\cC_{h,t,s}$ requires more demanding computations and is left as open problems (see \cite[Section 5]{neri_santonastaso_zullo}). In this article we will face with these problems, extending Theorem \ref{th:equivalencebetween_NSZ}. Our main achievement, Theorem \ref{main_result},  is obtained by heavily exploiting some concepts and results stated in \cite{longobardi-zanella_2022}. For this reason, we dedicated Section $2$, to recall basic definitions and salient aspects of these tools. As a direct consequence of the results stated in Section $3$, same arguments as those used in \cite{longobardi_marino_trombetti_zhou} and in   \cite{neri_santonastaso_zullo}, lead to the determination of the number of inequivalent codes of type ${\cal C}_{h,t,s},$ for $t \in \{3,4\}$. Finally, in Section $4$, we give a direct proof of the fact that the maximum scattered linear set of $\PG(1,q^8)$ associated with $\cC_{h,4,s}$ is not equivalent to any one known so far.

\section{Preliminaries}

In \cite{Csajbok_Marino_Polverino_Zanella}, it is pointed out that linear MRD codes with minimum distance $d=n-1$ that are $\F_{q^n}$-subspaces of $\mathcal{\tilde{L}}_{n,q}[x]$ can be seen, up to equivalence, as a suitable $2$-dimensional $\F_{q^n}$-subspace of $\mathcal{\tilde{L}}_{n,q}[x]$ of the following shape 
\[{\cal C}_f=\langle x,f(x) \rangle_{\F_{q^n}},\] where $f(x) \in \tilde{{\mathcal L}}_{n,q}[x]$ has the property that for each $y\in \F_{q^n}^*$ and $\kappa \in \F_{q^n}$, if $f(\kappa y)=\kappa f(y),$ then necessarily $\kappa \in \F_q$. Hereafter, without any loss of generality, we will suppose that the coefficient of the $q^0$-degree term in $f(x)$ is zero (see \cite{Csabok_Marino_Polverino_Zhou}). In \cite{sheekey_2015} such a $q$-polynomial is called  {\it scattered polynomial} and the $\F_q$-linear $n$-dimensional vector subspace of $\F_{q^n}^2$ $$U_f=\{(x,f(x)) \,:\, x \in \F_{q^n}\},$$ is called {\it scattered subspace}. In the following we will denote by $L_f$ the set of points of the projective line $\PG(1,q^n) \cong \PG(\F_{q^n}^2,\F_{q^n}),$ defined by the non-zero vectors in $U_f$, i.e. we put  \[ L_f =\{\langle (x,f(x)) \rangle_{\F_{q^n}} \,:\, x \in \F^*_{q^n}\},\] and will refer to it as to the linear set associated with the code ${\cal C}_f$. 

In \cite{exceptionalpartially}, two scattered polynomials $f(x)$ and $g(x) \in \tilde{{\mathcal L}}_{n,q}[x],$ are said to be $\GaL$-{\it equivalent} ($\GL$-{\it equivalent}) if $U_f$ and $U_g$ are $\GaL$-equivalent ($\GL$-equivalent); i.e., if there exists a matrix $ \begin{pmatrix} a & b \\ c & d  \end{pmatrix} \in \GL(2, \mathbb{F}_{q^n})$, and for each $y \in \F_{q^n}$ there exists $z \in \F_{q^n}$ with the property that \[\begin{pmatrix} a & b \\ c & d  \end{pmatrix} \begin{pmatrix} y^{\sigma} \\ f(y)^{\sigma}  \end{pmatrix} = \begin{pmatrix} z \\ g(z)  \end{pmatrix},\] where $\sigma \in {\rm Aut}(\F_{q^n})$ ($\sigma =id,$ in the case where $U_f$ and $U_g$ are $\GL$-equivalent). 
Moreover, in \cite{sheekey_2015}, it is shown that  ${\cal C}_f$ and ${\cal C}_g$ are equivalent if and only if $f$ and $g$ are $\GaL$-equivalent. Furthermore, if ${\cal C}_f$ and ${\cal C}_g$ are equivalent then, $L_f$ and $L_g$ are $\PGaL$-equivalent.
Although, the converse of this statement is not true: for instance, let $f(x)= x^q$ and $g(x)=x^{q^s}$ with $(s,n)=1$ and $s \not \equiv \pm 1 \pmod n$, then $f(x)$ and $g(x)$ are not $\GaL$-equivalent, but obviously $L_f=\{\langle (1,x^{q-1}) \rangle_{\F_{q^n}} \,:\, x \in \F_{q^n}\setminus\{0\}\}=L_g$.\\ 
\noindent Let $s$ be an integer such that $(s,n)=1$ and let $\eta \in \F_{q^n}$ such that ${\rm N}_{q^n / q} (\eta)=\eta^{\frac{q^{n}-1}{q-1}} \not \in \{0, 1\}$. 

\medskip
\noindent Until 2018 there were only two classes of MRD codes of $\tilde{{\mathcal L}}_{n,q}[x],$ with minimum distance $n-1$. Precisely

\begin{equation}\label{laveauw}
 {\mathcal G}_{2,s}=\langle x, x^{q^s} \rangle_{\F_{q^n}},
\end{equation} and

\begin{equation}\label{sheekey}
{\mathcal H}_{2,s,\eta} =\langle x, \eta x^{q^{s}}+x^{q^{(n-1)s}}\rangle_{\F_{q^n}}.
\end{equation}

\medskip
\noindent These are particular cases of two important families, existing for each assigned values of the minimum distance. In fact, the latter belongs to the class of codes known as {\it generalized twisted Gabidulin codes}, which were constructed by John Sheekey, in \cite{sheekey_2015}, while the former lies in the celebrated class of Generalized Delsarte-Gabidulin codes, which appeared in \cite[Sec.VI]{Roth}  (see also \cite{delsarte_bilinear_1978}, \cite{gabidulin_MRD_1985} and \cite{kshevetskiy_gabidulin_2005}, to go through the genesis of this construction).   

In 2018, a new example of such type of codes was found, in the case when $n=6$ and $q$ is odd. Precisely

\[{\cal Z}_{6,\zeta} = \langle x, x^q + x^{q^3}+ \zeta x^{q^5}\rangle_{\F_{q^6}},\] where here $\zeta \in \F_{q^6}^*$ such that $\zeta^2+\zeta= 1;$ see \cite{Csabok_Marino_Zullo} and \cite{Marino_Montanucci_Zullo}.

In addition, when $n\in \{6,8\}$, another construction was described in \cite{Csajbok_Marino_Polverino_Zanella}:   

\[{\cal K}_{n,s,\delta}=\langle x,\delta x^{q^s}+x^{q^{s
+n/2}}\rangle_{\F_{q^n}},\] where $(s,n/2)=1$, $\delta \in \F_{q^n}$ such that $\N_{q^n/q^{n/2}}(\delta) \not \in \{0,1\}$, and for some other conditions on $\delta$ and $q$.

\medskip
\noindent In \cite{Bartoli_Zanella_Zullo}, the authors provided a new family for $q$ odd, ensuing from the scattered $q$-polynomial \[f(x)=  h^{q-1}x^{q}-h^{q^2-1}x^{q^2}+x^{q^4}+x^{q^{5}},\] where $h\in \F_{q^6}$ such that $\N_{q^6/q^3}(h)=-1$. More precisely, they showed that elements of this family are new except when $q$ is a power of $5$, and $h \in \F_q$, in which case the code is equivalent to ${\cal Z}_{6,\delta}$.

Also in the case $q$ odd, in \cite{longobardi_zanella}, another new family have been found, existing for each $n=2t$ with either $t\geq3$ odd and $q\equiv 1 \pmod 4$, or $t$ even. Precisely, the following
\begin{equation}\label{longobardi_zanella}
\mathcal{C}_t = \langle x,  x^q+x^{q^{t-1}} -x^{q^{t+1}}+x^{q^{2t-1}}\rangle_{\F_{q^n}},
\end{equation}

\medskip
\noindent which was later generalized in \cite{longobardi_marino_trombetti_zhou} to  \begin{equation}\label{longobardi_marino_trombetti_zhou}
	{\mathcal C}_{h,t} = \langle x,\psi_{h,t}(x) \rangle_{\F_{q^{n}}}
	\end{equation}
where $\psi_{h,t}(x)=  x^{q}+x^{q^{t-1}}-h^{1-q^{t+1}}x^{q^{t+1}}+h^{1-q^{2t-1}}x^{q^{2t-1}}$ and $h \in \F_{q^{2t}} \setminus \F_{q^t}$ such that $\mathrm{N}_{q^{2t}/q^t}(h)=-1$. 
Performing same computations as in \cite[Theorem 4.6]{longobardi_marino_trombetti_zhou}, it is possible to see that  when $t=3$ the code ${\mathcal C}_{h,3}$ is equivalent to the one introduced in \cite{Bartoli_Zanella_Zullo}. Class (\ref{longobardi_marino_trombetti_zhou}) was further extended 
into the following:
\begin{equation}\label{neri_santonastaso_zullo} {\mathcal C}_{h,t,s} = \langle x,\psi_{h,t,s}(x)\rangle_{\F_{q^{2t}}},
\end{equation} where
$\psi_{h,t,s}(x)=x^{q^s}+x^{q^{s(t-1)}}  + h^{1+q^s}x^{q^{s(t+1)}}+h^{1-q^{s(2t-1)}}x^{q^{s(2t-1)}} \in \mathcal{ \tilde{L}}_{2t,q}[x]$, $(2t,s)=1$ and $h\in \F_{q^{2t}}$ such that $\mathrm {N}_{q^{2t}/q^t}(h)=-1$.

\medskip
\noindent In order to extend Theorem 1.1 stated in the Introduction by removing the hypothesis $t \geq 5$, as a first step  we recap and adapt some concepts and tools which were introduced in \cite[Section 4]{longobardi-zanella_2022}, and that will be crucial to perform subsequent computations.

Following \cite{longobardi-zanella_2022}, we say that a $q$-polynomial $F(x)=\sum_{i=0}^{n-1}c_ix^{q^i}$ is in {\it standard form}, if the greatest common divisor $m_F$ of the set of integers \[\{(i-j)\pmod n \, \colon \, c_ic_j\neq0 \text{ with } \, i \neq j\}\cup\{n\},\] is strictly larger than $1$. If this is the case, then $F(x)$ has the following fashion:
\begin{equation}\label{eq:sf}
    F(x)=\sum_{j=0}^{n/m-1}c_jx^{q^{s+jm}},
\end{equation}
where $m=m_F$ and $0\le s<m_F$. We stress here that when $F(x)$ is a scattered polynomial, then $s$ must be coprime with $m_F$, otherwise $F(x)$ would be $\mathbb{F}_{q^r}$-linear for $r=(s,m_F)>1$, 
contradicting its scatteredness.

\bigskip 
Let $\cC \leq \tilde{\mathcal{L}}_{n,q}[x]$ be a linear rank-metric code. Then, its \textit{left idealizer} and
\textit{right idealizer} are defined as the sets
\begin{equation*}
    I_L(\cC) = \{\varphi(x) \in \mathcal{\tilde{L}}_{n,q}[x] : \varphi \circ f \in \cC \,\, \forall f \in \cC\},
\end{equation*} 
and
\begin{equation*}
I_R(\cC) = \{\varphi(x) \in \mathcal{\tilde{L}}_{n,q}[x] :  f \circ \varphi \in \cC \,\, \forall f \in \cC\},
\end{equation*}
respectively. By \cite[Corollary 5.6]{LuTrZh17}, if $\cC$ is a linear MRD code, $I_L(\cC)$ and $I_R(\cC)$ are subalgebras of $\mathcal{\tilde{L}}_{n,q}[x]$; more precisely, they are isomorphic
to  subfields of $\F_{q^n}$. In particular, if $\cC$ has left idealizer containing a subfield isomorphic to $\F_{q^n}$, then it is called $\F_{q^n}$-\textit{linear}. Note that any MRD code $\cC_f$ associated with a scattered polynomial $f(x)$ is $\F_{q^n}$-linear.
The idealizers are useful tools  for studying the equivalence issue in this context; indeed MRD codes with non-isomorphic left (resp. right) idealizer cannot be equivalent, see \cite[Section 4]{LuTrZh17}.\\
Now, let $f(x)$ be a  scattered   $q$-polynomial and consider $G_f$ the stabilizer in $\GL(2,q^n)$ of the scattered subspace $U_f$. By \cite[Lemma 4.1]{longobardi_marino_trombetti_zhou}, (see also  \cite[Proposition 3.1]{longobardi-zanella_2022}) $G_f^\circ=G_f \cup \{O\}$, where $O$ is the null matrix of order 2, is a subfield of matrices isomorphic to $I_R(\cC_f)$ as field. Then, one gets the following
\begin{theorem}\label{t:sf}
Let ${\cal C}_F=\langle x,F(x) \rangle_{\F_{q^n}} \subset \tilde{\mathcal{L}}_{n,q}[x]$ be a linear $2$-dimensional MRD code with minimum distance $n-1$.
The following statements are equivalent:
\begin{enumerate}
\item [$(i)$] $|G^\circ_F|=q^m$, $m > 1$, and all elements of $G^\circ_F$ are diagonal.
\item [$(ii)$] the $q$-polynomial $F(x)$ is in standard form.
\item [$(iii)$] the right idealizer $$I_R(\cC_F)=\{\alpha x \,:\, \alpha \in \F_{q^{m}}\}.$$

\end{enumerate}
Moreover, if the above conditions $(i)$, $(ii)$ and $(iii)$ hold, then $m = m_F$ and
\begin{equation}
    G^\circ_F=\left \{\begin{pmatrix}
    \alpha & 0 \\
     0 & \alpha^{q^s}   \end{pmatrix} \colon \alpha \in \F_{q^m} \right \},
\end{equation}
with $(s,m)=1$.
\end{theorem}
\begin{proof}
By \cite[Theorem 4.3.]{longobardi-zanella_2022}, $(i)$ is equivalent to $(ii)$.\\
$(ii) \Rightarrow (iii)$. Since $F(x)$ is  as in Formula \eqref{eq:sf}, the set
\[\{ \alpha x : \alpha \in \F_{q^m}\}\]
is contained in $I_R(\cC_F)$. Taking into account $|G^\circ_F|=|I_R(\cC_F)|=q^m$, $(iii)$ follows.\\
$(iii) \Rightarrow (i)$. Let $\omega \in \F_{q^m}$ such that $\F_q(\omega)=\F_{q^m}$. Since $\omega x \in I_R(\cC_F)$, then there exists $d \in \F_{q^n}$ such that
\begin{equation}\label{df(x)}
   F(\omega x)=dF(x). 
\end{equation}
This implies that if $F(x) = \sum^{n-1}_{i=1}c_ix^{q^i}$, then \[c_i\omega^{q^i}= dc_i\]
for any $i =1, \ldots, n -1$.
So, if $c_i \not = 0$, $d \in \F_{q^m}$ and $d=\omega^h$ for some $h \in \mathbb{N}$.
By \eqref{df(x)}, the matrix
\begin{equation*}
\begin{pmatrix}
\omega & 0 \\
0 & \omega^h
\end{pmatrix}
\end{equation*}
 belongs to $G_F$, has order $q^m-1$ and, since $|G_F|=|I_R(\cC_F) \setminus \{0\}|=q^m-1$, it generates $G_F$. This implies $(i)$.
\end{proof}

As a direct consequence of \cite[Corollary 4.7] {longobardi-zanella_2022}, one has that if the MRD code ${\cal C}_f =\langle x,f(x) \rangle_{\F_{q^n}}$ has right idealizer $I_R({\cal C}_f)$ not isomorphic to $\F_q$ then, ${\cal C}_f$ is equivalent to an MRD code ${\cal C}_F$, where $F$ is a $q$-polynomial shaped as in (\ref{eq:sf}). In addition, as a consequence of what stated in \cite[Section 4]{longobardi-zanella_2022}, we can prove the following result concerning with the equivalence issue for codes in the form ${\cal C}_F =\langle x,F(x) \rangle_{\F_{q^n}}$ where $F(x)$ is a $q$-polynomial in standard form. Precisely,

\begin{theorem}\label{diagonal}
Let ${\cal C}_{F_i}$, $i=1,2$ be two $2$-dimensional MRD codes where $F_i$,  $i=1,2$, are scattered polynomials having the form described in (\ref{eq:sf}). Then, they are equivalent if and only if there exist $a,b,c,d \in \F_{q^n}^*$ such that
\begin{equation}\label{eq-standard}
    dF_2(x)=F_1^\rho(ax) \quad \textnormal{or} \quad F_1^\rho(bF_2(x))=cx 
\end{equation}
for some $\rho \in \Aut(\F_{q^n})$.
In particular, 
\begin{itemize}
    \item [$(i)$] $\cC_{F_2}=\cC^\rho_{F_1} \circ ax$ or,
    \item [ $(ii)$] $\cC_{F_2}=\cC^\rho_{F_1} \circ b F_2(x)$.
\end{itemize}
\end{theorem}
\begin{proof}
By \cite[Theorem 8]{sheekey_2015}, $\cC_{F_1}$ and $\cC_{F_2}$ are equivalent codes if and only if $F_1(x)$ and $F_2(x)$ are $\GaL$-equivalent, or equivalently, $F_1^{\rho}(x)$ and $F_2(x)$ are $\GL$-equivalent for some $\rho \in \Aut(\F_{q^n})$. Then, there exists a matrix  $A=\begin{pmatrix} a & b \\ c & d  \end{pmatrix} \in \GL(2,q^n)$ such that for any $y \in \F_{q^n}$, there is $z \in \F_{q^n}$ such that
\begin{equation}\label{F-F'}
    A\begin{pmatrix} y \\ F_2(y) \end{pmatrix} = \begin{pmatrix} z \\ F_1^\rho(z) \end{pmatrix}.
\end{equation}  Let $G_{F_1^\rho}$ and $G_{F_2}$ be the stabilizer groups in $\GL(2,q^n)$ of the subspaces $U_{F_1^\rho}$ and $U_{F_2}$, respectively.
It is straightforward to see that 
\begin{equation}\label{eq:conjugation}
    A^{-1}G_{F_1^\rho}A=G_{F_2}.
\end{equation}

By Theorem \ref{t:sf}, $G_{F_1^\rho}$ and $G_{F_2}$ are groups made up of diagonal matrices. Therefore, by Formula \eqref{eq:conjugation} one gets that $a=d=0$ or $b=c=0$.\\
If $a=d=0$, $b,c \in \F_{q^n}^*$ and  $c x=F_1^\rho(b F_2(x))$. So, one gets the second formula in \eqref{eq-standard}.\\
If $b=c=0$, $a,d \in \F_{q^n}^*$ and, by $\eqref{F-F'}$, $dF_2(x)=F_1^\rho(ax)$.\\
Now, if $dF_2(x)=F_1^\rho(ax)$,  
\begin{equation}
\begin{split}
    \cC_{F_2}&=\langle x , F_2(x) \rangle_{\F_{q^n}}=\langle x, dF_2(x) \rangle_{\F_{q^n}} =\langle x, F_1^\rho(ax) \rangle_{\F_{q^n}}\\
    &= \langle x, F_1^\rho(x) \rangle_{\F_{q^n}}\circ ax= \cC^\rho_{F_1} \circ ax.
    \end{split}
\end{equation}
If $F_1^\rho(bF_2(x))=cx$, then
\begin{equation}
\begin{split}
    \cC_{F_2}&=\langle x, F_2(x) \rangle_{\F_{q^n}}=\langle cx, F_2(x) \rangle_{\F_{q^n}}=\langle F_1^\rho(bF_2(x)), F_2(x) \rangle_{\F_{q^n}}\\
    &=\langle x, F_1^\rho(x) \rangle_{\F_{q^n}} \circ b F_2(x)= \cC^\rho_{F_1} \circ bF_2(x).
    \end{split}
\end{equation}
\end{proof}

Also, as consequence of \cite{longobardi-zanella_2022} and the Theorem above, we have the following

\begin{corollary}\label{eq-Snq}
Let $\cC_{f_i}$ be MRD codes such that  $I_R({\cal C}_{f_i}),$ are not isomorphic to $\F_q$,  $i=1,2$. Then, $\cC_{f_1}$ and $\cC_{f_2}$ are equivalent if and only if there exist two polynomials $F_i$ in standard forms, such that ${\cal C}_{f_i}$ is equivalent to $\cC_{F_i}$,  $i=1,2$, and \eqref{eq-standard} holds true.
\end{corollary}

 By \cite[Proposition 3.4, Remark 3.5 and Remark 4.4]{longobardi-zanella_2022}, one can see that for each $t \geq 3$ the right idealizer of ${\mathcal C}_{h,t,s}$ is isomorphic to $\F_{q^2}.$ In addition, any code ${\mathcal C}_{h,t,s} = \langle x,\psi_{h,t,s}(x)\rangle_{\F_{q^{2t}}}$ with $t \geq 4$ and even, is already expressed by means of a polynomial in standard form. On the other hand, by \cite[Example 4.11]{longobardi-zanella_2022}, one gets that the code $\cC_{H}$ where 
\begin{equation}
    H(x):=H_{h,s}(x)=(1-h^{1+q^{2s}})x^{q^s}+(h+h^2)x^{q^{3s}}+h^{1+q^{2s}}(h+h^{q^s})x^{q^{5s}} \in \mathcal{\tilde{L}}_{6,q}[x],
\end{equation}
 is MRD, equivalent to ${\mathcal C}_{h,3,s}$ and $H(x)$ is in standard form.

\section{On the equivalence issue for ${\mathcal C}_{h,t,s}$, with $t\in \{3,4\}$}

In the following, we assume $n=2t$ and restrict ourselves to the case $t\in\{3,4\}$ and start by studying the equivalence between two different elements in ${\cal C}_{h,t,s}$. Firstly, we prove the following two technical Lemmas.

\begin{lemma}\label{lm:case_t_=3}
Assume $n=6$ and that $\mathcal{C}_{h,3,s}$ and \,$\mathcal{C}_{k,3,\ell}$, are equivalent. One gets the following:
\begin{itemize}
\item[\textnormal{1.}] if $\ell \equiv s \pmod 6$, then $h^{\rho}=\pm k,$
\item  [\textnormal{2.}] f $\ell \equiv -s \pmod 6$, then $h^{\rho}=\pm k^{-1},$
\end{itemize}where $\rho \in \Aut(\F_{q^6})$.
\end{lemma}

\begin{proof}
Assume that $\mathcal{C}_{h,3,s}=\langle x,\psi_{h,3,s}(x) \rangle_{\F_{q^6}}$ and\, $\mathcal{C}_{k,3,\ell}=\langle x,\psi_{k,\ell,s}(x) \rangle_{\F_{q^6}},$ are equivalent. By Corollary \ref{eq-Snq}, this may happen if and only if Formula \eqref{eq-standard} holds true for any pair of standard forms of the $q$-polynomials $\psi_{h,3,s}$ and $\psi_{k,3,\ell}$, respectively. As we indicated in the last part of the previous section, in \cite{longobardi-zanella_2022}, it was showed that $\psi_{h,3,s}(x)$ is equivalent to the following polynomial in standard form:
\[H_{h,s}(x)=(1-h^{1+q^{2s}})x^{q^s}+(h+h^2)x^{q^{3s}}+h^{1+q^{2s}}(h+h^{q^s})x^{q^{5s}}.\] 

Following Corollary \ref{eq-Snq}, to investigate this equivalence issue we may consider the following two equations 
\begin{equation}\label{eq:equiv_t=3}
		dH_{k,\ell}(x)=H^{\rho}_{h,s}(ax), \quad \textnormal{and} \quad H_{h,s}^\rho\left( b H_{k,\ell}(x)\right)=cx
\end{equation}
for $a,b,c,d \in \F_{q^6}^*$ and $\rho \in \Aut(\F_{q^6})$. 
Since the automorphism $\rho$ acts on $h$, without loss of generality, we may suppose that it is the identity, see \cite[Remark 4.5]{neri_santonastaso_zullo}. 
\medskip

\noindent \textbf{Case 1.} Suppose $s \equiv \ell \pmod 6$. If 
$dH_{k,\ell}(x)=H_{h,s}(ax)$, by expanding this equation we get the following conditions
\begin{equation}\label{a-d-systems=l}
    \begin{cases}
    d(1-k^{1+q^{2s}})=a^{q^s}(1-h^{1+q^{2s}})\\
    d(k+k^2)=a^{q^{3s}}(h+h^2)\\
    dk^{1+q^{2s}}(k+k^{q^s})=a^{q^{5s}}h^{1+q^{2s}}(h+h^{q^s}).
    \end{cases}
\end{equation}
We focus on getting a relation between $h$ and $k$ and eventually the values of $a$ and $d$. In order to do this we start by considering the ratio between second and first equation. Doing so we get
\begin{equation}\label{eq:case_1_1}
a^{q^{3s}-q^s}=\frac{k(k+1)(1-h^{1+q^{2s}})}{h(h+1)(1-k^{1+q^{2s}})}.
\end{equation}

Similarly computing the ratio between third and second equation and between first and third in \eqref{a-d-systems=l} we get
\begin{equation}\label{eq:case_1_2}
\begin{split}
&a^{q^{5s}-q^{3s}}=\frac{k^{q^{2s}}(h+1)(k+k^{q^s})}{h^{q^{2s}}(k+1)(h+h^{q^s})}\\
    &a^{q^s-q^{5s}}=\frac{(1-k^{1+q^{2s}})(h+h^{q^s})h^{1+q^{2s}}}{(1-h^{1+q^{2s}})(k+k^{q^s})k^{1+q^{2s}}}.
    \end{split}
    \end{equation}
By rising Equation (\ref{eq:case_1_1}) and the two Equations in (\ref{eq:case_1_2}) to the $q^{2s}$-power, and taking into account that $h^{q^{3s}+1}=k^{q^{3s}+1}=-1$, we get the following system of conditions on $h$ and $k$:

\begin{equation}
    \begin{cases}
    \left(\frac{k}{h}\right)^{q^{5s}}= \left [\frac{(k^{q^{4s}}+k^{q^{5s}})(h+1)}{(h^{q^{4s}}+h^{q^{5s}})(k+1)} \right] ^{q^{2s}+1}\\
   \left( \frac{k}{h} \right )^{q^{2s}+q^{4s}}=\frac{(h+h^{q^{s}})(k^{q^{2s}}+1)}{(k+k^{q^{s}})(h^{q^{2s}}+1)}\\
    \left( \frac{k}{h} \right )^{q^{4s}}=\frac{(h+1)(k^{q^{5s}}+k^{q^{4s}})}{(k+1)(h^{q^{5s}}+h^{q^{4s}})}.
    \end{cases}
\end{equation}
Putting $\xi=\frac{(k^{q^{4s}}+k^{q^{5s}})(h+1)}{(h^{q^{4s}}+h^{q^{5s}})(k+1)}$, then the system above reads

\begin{equation*}
      \begin{cases}
    \left(\frac{k}{h}\right)^{q^{5s}}= \xi ^{q^{2s}+1}\\
   \left( \frac{k}{h} \right )^{q^{2s}+q^{4s}}=\frac{1}{
\xi^{q^{2s}}}\\
    \left( \frac{k}{h} \right )^{q^{4s}}=\xi.
    \end{cases}
\end{equation*}
Since $\xi^{1+q^{3s}}=1$, then
\begin{equation}\label{(20)}
      \begin{cases}
    \left(\frac{k}{h}\right)^{q^{5s}}= \xi ^{q^{2s}+1}\\
   \left( \frac{k}{h} \right )^{q^{2s}+q^{4s}}=\xi^{q^{5s}}\\
    \left( \frac{k}{h} \right )^{q^{4s}}=\xi.
    \end{cases}
\end{equation}
Now, raising the third equation of \eqref{(20)} the $q^s$-power and taking into account the first one, we get that $\xi^{q^{2s}-q^{s}+1}=1$, and so the system becomes
\begin{equation}\label{syst_case_1}
      \begin{cases}
    \left(\frac{k}{h}\right)^{q^{5s}}= \xi ^{q^{s}}\\
   \left( \frac{k}{h} \right )^{q^{2s}+q^{4s}}=\xi^{q^{5s}}\\
    \left( \frac{k}{h} \right )^{q^{4s}}=\xi.
    \end{cases}
\end{equation}
From the first equation of above System (\ref{syst_case_1}) and taking into account the expression of $\xi$, we have
\begin{equation*}
    \xi=\frac{(\xi h^{q^{4s}}+\xi^{q^s}h^{q^{5s}})(h+1)}{(h^{q^{4s}}+h^{q^{5s}})(\xi^{q^{2s}}h+1)}.
\end{equation*}
From latter formula we obtain
\begin{equation*}
    (\xi^{q^{s}}-\xi)h^{q^{4s}+1}= (\xi^{q^{s}}-\xi)h^{q^{5s}}.
\end{equation*}
Now if $\xi \not \in \F_{q}$ (that with condition $\xi^{q^{2s}-q^s+1}=1$ implies $\xi \neq 1$), then $h^{q^{4s}+1}=h^{q^{5s}}$, a contradiction, \cite[Lemma 2.3\,$(3)$]{Bartoli_Zanella_Zullo}. Hence $\xi=1$, $h=k$ and, by system \eqref{a-d-systems=l}, $a \in \F_{q^2}$ and $d=a^{q}$.

\medskip
\noindent If otherwise $H_{h,s}(bH_{k,s}(x))=cx,$ then by expanding this formula we get
\begin{align*}
\hspace{-0.8cm}&\left(k^{q^s+q^{3s}}(1-h^{1+q^{2s}})(k^{q^s}+k^{q^{2s}})b^{q^{s}}+(h+h^2)(k^{q^{3s}}+k^{2q^{3s}})b^{q^{3s}}+h^{1+q^{2s}}(h+h^{q^s})(1-k^{q^{5s}+q^{s}})b^{q^{5s}}\right)x\\
&+\left((1-h^{1+q^{2s}})(k^{q^s}+k^{2q^{s}})b^{q^{s}}+(h+h^2)(1-k^{q^{3s}+q^{5s}})b^{q^{3s}}+k^{q^s+q^{5s}}h^{1+q^{2s}}(h+h^{q^s})(k^{q^{5s}}+k)b^{q^{5s}}\right)x^{q^{4s}}\\
&+\left((1-h^{1+q^{2s}})(1-k^{q^s+q^{3s}})b^{q^{s}}+k^{q^{3s}+q^{5s}}(h+h^2)(k^{q^{3s}}+k^{q^{4s}})b^{q^{3s}}+h^{1+q^{2s}}(h+h^{q^s})(k^{q^{5s}}+k^{2q^{5s}})b^{q^{5s}}\right)x^{q^{2s}}\\
&=cx.
\end{align*}
Now, comparing the coefficients of $x, x^{q^{2s}}$ and $x^{q^{4s}}$ on the left and right side, and taking into account that $h^{q^{3s}+1}=k^{q^{3s}+1}=-1$, we obtain the following linear system in the unknowns $b^{q^s},b^{q^{3s}}$ and $b^{q^{5s}}$:
\begin{equation}\label{LSt=3}
A_{h,k,s} \begin{pmatrix}
    b^{q^s} \\
    b^{q^{3s}} \\
    b^{q^{5s}} 
    \end{pmatrix}=
    \begin{pmatrix}
    c \\
    0 \\
    0 
    \end{pmatrix}
\end{equation}
where
\begin{equation}
\hspace{-0.5cm}A_{h,k,s}=\begin{pmatrix}
    k^{q^s+q^{3s}}(1-h^{1+q^{2s}})(k^{q^s}+k^{q^{2s}}) & (h+h^2)(k^{q^{3s}}+k^{2q^{3s}})  & h^{1+q^{2s}}(h+h^{q^s})(1-k^{{q^{5s}}+q^s})\\
(1-h^{1+q^{2s}})(k^{q^s}+k^{2q^s}) & (h+h^2)(1-k^{q^{3s}+q^{5s}}) & k^{q^{5s}+q^s}h^{1+q^{2s}}(h+h^{q^s})(k^{q^{5s}}+k)\\
(1-h^{1+q^{2s}})(1-k^{q^s+q^{3s}}) &  k^{q^{3s}+q^{5s}}(h+h^2)(k^{q^{3s}}+k^{q^{4s}}) & h^{1+q^{2s}}(h+h^{q^s})(k^{q^{5s}}+k^{2q^{5s}})
    \end{pmatrix}.
\end{equation}
    The determinant of $A_{h,k,s}$ is 
  
    \begin{equation}
     h^{1+q^{2s}}(h^{1+q^{2s}}-1) (h+h^2)(h+h^{q^s})  (k^{q^{s}+q^{3s}+q^{5s}}+1)^3. \end{equation}
   
 First we see that this expression cannot be equal to zero. By way of contradiction, let us assume that this is the case. We can easily see that none of the terms containing $h$ can be equal to zero, because this is in contradiction with either the fact that $h$ is non zero or that $h^{q^{3s}+1}=-1$. The term with $k$ needs more attention. But note that if it is zero, then
 \begin{equation*}
     k^{q^{s}+q^{3s}+q^{5s}}=-1.
 \end{equation*}
Rising the above equation to the $q^{3s}$ power and then taking ratio with $k^{1+q^{3s}}=-1$ we get
 \begin{equation*}
   k^{q^{2s}-q^s+1}=1 
 \end{equation*}
 which finally leads to a contradiction as $-1=k^{1+q^{3s}}=(k^{q^{2s}-q^s+1})^{q^s+1}=1.$  Hence, the determinant of $A_{h,k,s}$ is different from zero. \\ 
 \noindent As by the assumption $c\neq 0$ the unique solution of the System (\ref{LSt=3}) is ${\bf b}=\big (b^{q^s},b^{q^{3s}},b^{q^{5s}}\big ) \in \mathbb{F}_{q^6}^3,$ with
\begin{equation}\label{solutiont=3}
\begin{split}
   &b^{q^s}=\frac{ck^{q^{5s}}(1-k^{q^{3s}+q^{5s}})}{(h^{1+q^{2s}}-1)(k^{q^{s}+q^{3s}+q^{5s}}+1)^2}\\
&b^{q^{3s}}=\frac{-ck^{1+q^s+q^{5s}}(k^{q^{3s}}+1)}{(h+h^2)(k^{q^{s}+q^{3s}+q^{5s}}+1)^2}\\
&b^{q^{5s}}=\frac{c(1-k^{{q^s}+q^{3s}})}{ h^{1+q^{2s}}(h+h^{q^s})(k^{q^{s}+q^{3s}+q^{5s}}+1)^2}.
 \end{split}
\end{equation}

Also in this case, we focus on getting a relation between $h$ and $k$ and eventually the  values of $b$ and $c$ satisfying the second equation expressed in (\ref{eq:equiv_t=3}). In order to do this we start by computing the ratio between second and first equation in (\ref{solutiont=3}). From this we get
\begin{align}
b^{q^{3s}-q^s}=-\frac{k^{1+q^{s}}(k^{q^{3s}}+1)(h^{1+q^{2s}}-1)}{h(h+1)(1-k^{q^{3s}+q^{5s}})}. \label{firstratio}
\end{align}
Similarly, taking the ratio of third and second, and of the first and the third equation in (\ref{solutiont=3}) we get
\begin{align}
   &b^{q^{5s}-q^{3s}}=-\frac{(1-k^{q^s+q^{3s}})(h+1)}{h^{q^{2s}}(h+h^{q^s})k^{1+q^s+q^{5s}}(k^{q^{3s}}+1)} \label{secondratio}, \\ 
   & b^{q^{s}-q^{5s}}=\frac{h^{1+q^{2s}}(h+h^{q^s})k^{q^{5s}}(1-k^{q^{3s}+q^{5s}})}{(1-k^{q^s+q^{3s}})(h^{1+q^{2s}}-1)} \label{thirdratio},
\end{align} respectively. 
Since the $q^{2s}$-power of Formulas \eqref{firstratio},\eqref{secondratio} and \eqref{thirdratio} returns Formulas \eqref{secondratio}, \eqref{thirdratio} and \eqref{firstratio}, respectively, we get
\begin{equation}\label{rhosystem}
\begin{cases}
k^{q^s}h^{q^{2s}}=\left [\frac{(1-k^{q^s+q^{5s}})(h+1)}{(h^{q^{4s}}+h^{q^{5s}})(k^{q^{3s}}+1)}\right]^{q^{2s}+1}\\
k^{q^{3s}}h^{q^{4s}}=\frac{(k^{q^{3s}}+1)(h^{q^{4s}}+h^{q^{5s}})}{(h+1)(1-k^{q^{5s}+q^{s}})}\\
k^{q^s+q^{3s}}h^{q^{2s}+q^{4s}}=\frac{(1-k^{q^s+q^{3s}})(h^{q^{2s}}+1)}{(k^{q^{5s}}+1)(h+h^{q^s})}.
\end{cases}
\end{equation}
Putting $\xi=\frac{(1-k^{q^s+q^{5s}})(h+1)}{(h^{q^{4s}}+h^{q^{5s}})(k^{q^{3s}}+1)}$, conditions above reads
\begin{equation}\label{rhosystem2}
\begin{cases}
k^{q^s}h^{q^{2s}}=\xi^{q^{2s}+1}\\
k^{q^{3s}}h^{q^{4s}}=1/\xi\\
k^{q^s+q^{3s}}h^{q^{2s}+q^{4s}}=\xi^{q^{2s}}.
\end{cases}
\end{equation}
By the first and the second equation, one gets that $\xi^{q^{4s}+q^{2s}+1}=1$ and since $h^{1+q^{3s}}=k^{1+q^{3s}}=-1$, $\xi^{q^{3s}+1}=1$ and so $\xi^{q^{2s}-q^{s}+1}=1$. Therefore System \eqref{rhosystem2} becomes
\begin{equation}\label{rhosystem3}
\begin{cases}
k^{q^s}h^{q^{2s}}=\xi^{q^s}\\
k^{q^{3s}}h^{q^{4s}}=\xi^{q^{3s}}\\
k^{q^s+q^{3s}}h^{q^{2s}+q^{4s}}=\xi^{q^{2s}}.
\end{cases}
\end{equation}
From the first equation in System (\ref{rhosystem3}), and taking into account the expressions of $\xi$, we have
\begin{equation}
\begin{split}
    &\xi=\frac{(h^{q^{2s}+1}-\xi)(h+1)h^{q^{4s}}}{h^{q^{2s}+1}(h^{q^{4s}}+h^{q^{5s}})(\xi^{q^{3s}}+h^{q^{4s}})}\\
    &\xi(h^{1+q^{2s}+q^{4s}}+1)h^{q^{4s}}=h(h^{1+q^{2s}+q^{4s}}+1)\\
&\xi=h^{1-q^{4s}},
    \end{split}
\end{equation}
obtaining $h=-k$. 
Then, taking this into account in \eqref{solutiont=3}, we get

\begin{equation}
\begin{split}
   &b^{q^s}=\frac{ck^{q^{5s}}(1-k^{q^{3s}+q^{5s}})}{(k^{1+q^{2s}}-1)(k^{q^{s}+q^{3s}+q^{5s}}+1)^2}\\
&b^{q^{3s}}=-\frac{ck^{q^s+q^{5s}}(k^{q^{3s}}+1)}{(k-1)(k^{q^{s}+q^{3s}+q^{5s}}+1)^2}\\
&b^{q^{5s}}=-\frac{c(1-k^{{q^s}+q^{3s}})}{ k^{1+q^{2s}}(k+k^{q^s})(k^{q^{s}+q^{3s}+q^{5s}}+1)^2}.
 \end{split}
\end{equation}
Finally, by raising one of these three equations to the $q^{2s}$-power and comparing it with the other ones, we obtain 
\[ c=\lambda k^{q^{2s}+1},\]
for some $\lambda \in \F_{q^2}$, and so $b=\frac{\lambda^{q}k^{q^{4s}}}{(k^{1+q^{2s}+q^{4s}}+1)^2}$.\\

\noindent 
\textbf{Case 2.} Now, we assume $s \equiv -\ell \pmod 6$. If $dH_{k,-s}(x)=H_{h,s}(ax)$ then we get the following conditions
\begin{equation}\label{a-d-systembczero(-1)}
    \begin{cases}
    d(1-k^{1+q^{4s}})=a^{q^{5s}}h^{1+q^{2s}}(h+h^{q^s})\\
    d(k+k^2)=a^{q^{3s}}(h+h^2)\\
    d k^{1+q^{4s}}(k+k^{q^{5s}})=a^{q^s}(1-h^{1+q^{2s}}).
    \end{cases}
\end{equation}
As before we have
\begin{equation}
\begin{split}
    &a^{q^{3s}-q^s}=\frac{k(k+1)(1-h^{1+q^{2s}})}{h(h+1)(k+k^{q^{5s}})k^{1+q^{4s}}},\\
    &a^{q^{5s}-q^{3s}}=\frac{(1-k^{1+q^{4s}})(h+1)}{h^{q^{2s}}(h+h^{q^s})k(k+1)},\\
    &a^{q^s-q^{5s}}=\frac{(k+k^{q^{5s}})k^{1+q^{4s}}h^{1+q^{2s}}(h+h^{q^s})}{(1-k^{1+q^{4s}})(1-h^{1+q^{2s}})}.
 \end{split}
\end{equation}
Since the $q^{2s}$-power of first, second and third equation above, is second, third and first equation respectively, then,

\begin{equation}
    \begin{cases}
   k^{q^{3s}}h^{q^{2s}}= -\left [\frac{(k^{q^{2s}}+k^{q^{s}})(h+1)}{(h^{q^{5s}}+h^{q^{4s}})(k+1)} \right] ^{q^{2s}+1}\\
   k^{q^{4s}}h^{q^{2s}+q^{4s}}=-\frac{(k^{q^{3s}}+k^{q^{4s}})(h^{q^{2s}}+1)}{(h+h^{q^{s}})(k^{q^{2s}}+1)}\\
    k^{1+q^{4s}}h^{q^{4s}}=-\frac{(k+1)(h^{q^{4s}}+h^{q^{5s}})}{(h+1)(k^{q^{2s}}+k^{q^{s}})}
    \end{cases}
\end{equation}
Putting $\xi=\frac{(k^{q^{2s}}+k^{q^{s}})(h+1)}{(h^{q^{5s}}+h^{q^{4s}})(k+1)}$, then the system above becomes
\begin{equation}\label{11}
    \begin{cases}
    k^{q^{3s}}h^{q^{2s}}=-\xi^{q^{2s}+1}\\
    k^{q^{4s}}h^{q^{2s}+q^{4s}}=-\xi^{q^{2s}}\\
     k^{1+q^{4s}}h^{q^{4s}}=-\frac{1}{\xi}
    \end{cases}
\end{equation}
    From the last equation of System \eqref{11} we have $\xi^{q^{3s}+1}=-1$ and using the first and second equation,  we get $h^{1+q^{2s}+q^{4s}}=-\xi^{1+q^{2s}+q^{4s}}$. Now, taking into account this last relation between $h$ and $\xi$, and the first equation of \eqref{11} in the expression of $\xi$, we get
    \begin{equation}
        \begin{split}
            \xi&=\frac{(\xi^{q^{5s}+q^{s}}h^{q^{4s}}+\xi^{q^{4s}+1}h^{q^{3s}})(h+1)}{(\xi^{q^{5s}+q^{3s}}h^{q^{2s}}+1)(h^{q^{4s}}+h^{q^{5s}})}\\
(-\xi^{q^{5s}}h^{q^{2s}}+\xi)(h^{q^{4s}}+h^{q^{5s}})&=(\xi^{q^{5s}+q^{s}}h^{q^{4s}}+\xi^{q^{4s}+1}h^{q^{3s}})(h+1)\\
-\xi^{q^{5s}}h^{q^{2s}+q^{4s}}+\xi^{q^{5s}}+\xi h^{q^{4s}}+\xi h^{q^{5s}}&=\xi^{q^{5s}+q^s}h^{q^{4s}+1}+\xi^{q^{5s}+q^s}h^{q^{4s}}-\xi^{q^{4s}+1}+\xi^{q^{4s}+1}h^{q^{3s}}.        \end{split}
    \end{equation}
 Now as $h^{1+q^{2s}+q^{4s}}=-\xi^{1+q^{2s}+q^{4s}}$, $\xi^{q^{3s}+1}=-1$ and $h^{q^{3s}+1}=-1$, we have $\xi^{q^{5s}}h^{q^{2s}+q^{4s}}=-\xi^{q^{4s}+1}h^{q^{3s}}$ and $\xi^{q^{5s}+q^{s}}h^{q^{4s}+1}=\xi h^{q^{5s}}$. Therefore the above equation reduces to the following
    \begin{equation}
        \begin{split}
(\xi^{q^{5s}}+\xi^{q^{4s}+1}) &=(\xi^{q^{5s}+q^s}-\xi)h^{q^{4s}}\\
 (\xi^{q^{5s}+q^s}-\xi) &=\xi^{q^s}(\xi^{q^{5s}+q^s}-\xi)h^{q^{4s}}\\
    \xi&=-h   .
        \end{split}
    \end{equation}
    Using this in first equation of (\ref{11}) we have $h=k^{-1}$. Substituting this in second and third equation of (\ref{a-d-systembczero(-1)}) we get $d k^3=a^{q^{3s}}$ and $d k^{2+q^{4s}}=a^{q^s}$. Taking $q^{2s}$-power of the latter equation and taking it into account in the former, we get $d=\lambda k^{-2}$, for some $\lambda\in \mathbb{F}_{q^2}$ and hence $a=-\lambda^{q^s}k^{-1}$.\\

Finally, we take into consideration the case $H_{h,s}(bH_{k,-s}(x))=cx$. 
\noindent Expanding this equation we get 
\begin{align*}
&[(1-h^{1+q^{2s}})(1-k^{q^s+q^{5s}})b^{q^s}+(h+h^2)(k^{q^{3s}}+k^{2q^{3s}})b^{q^{3s}}+h^{1+q^{2s}}(h+h^{q^s})k^{q^{5s}+q^{3s}}(k^{q^{5s}}+k^{q^{4s}})b^{q^{5s}}]x\\
&+[(1-h^{1+q^{2s}})k^{q^s+q^{5s}}(k^{q^{s}}+k)b^{q^s}+(h+h^2)(1-k^{q^{3s}+q^{s}})b^{q^{3s}}+h^{1+q^{2s}}(h+h^{q^s})(k^{q^{5s}}+k^{2q^{5s}})b^{q^{5s}}]x^{q^{2s}}\\
&+[(1-h^{1+q^{2s}})(k^{q^s}+k^{2q^s})b^{q^s}+(h+h^2)k^{q^{3s}+q^{s}}(k^{q^{3s}}+k^{q^{2s}})b^{q^{3s}}+h^{1+q^{2s}}(h+h^{q^s})(1-k^{q^{5s}+q^{3s}})b^{q^{5s}}]x^{q^{4s}}\\
&=cx.
\end{align*} 
Now, by equating the coefficients of the terms $x, x^{q^{2s}}$ and $x^{q^{4s}}$ on the left and right side, and taking into account that $h^{q^{3s}+1}=k^{q^{3\ell}+1}=-1$, we obtain the following linear system in $b^{q^s},b^{q^{3s}}$ and $b^{q^{5s}}$:
\begin{equation}\label{LSt=3(-1)}
    A'_{h,k,s} \begin{pmatrix}
    b^{q^s} \\
    b^{q^{3s}} \\
    b^{q^{5s}} 
    \end{pmatrix}=
    \begin{pmatrix}
    c \\
    0 \\
    0 
    \end{pmatrix},
\end{equation}
where
\begin{equation*}
\hspace{-1cm}
A'_{h,k,s}=
\begin{pmatrix}
    (1-h^{1+q^{2s}})(1-k^{q^s+q^{5s}}) & (h+h^2)(k^{q^{3s}}+k^{2q^{3s}}) & k^{q^{5s}+q^{3s}}h^{1+q^{2s}}(h+h^{q^s})(k^{q^{5s}}+k^{q^{4s}})\\
k^{q^s+q^{5s}}(1-h^{1+q^{2s}})(k^{q^{s}}+k) &(h+h^2)(1-k^{q^{3s}+q^{s}}) & h^{1+q^{2s}}(h+h^{q^s})(k^{q^{5s}}+k^{2q^{5s}})\\
(1-h^{1+q^{2s}})(k^{q^s}+k^{2q^s}) & k^{q^{3s}+q^{s}}(h+h^2)(k^{q^{3s}}+k^{q^{2s}}) & h^{1+q^{2s}}(h+h^{q^s})(1-k^{q^{5s}+q^{3s}})
\end{pmatrix}
  \end{equation*}
Direct computation shows that ${\rm det}(A'_{h,k,s})={\rm det}(A_{h,k,s}) \neq 0$. Therefore, $A'_{h,k,s}$ is also non singular.
 \noindent As by the assumption $c\neq 0$, then the unique solution of the Linear System in (\ref{LSt=3(-1)}), is ${\bf b}=\big (b^{q^s},b^{q^{3s}},b^{q^{5s}}\big ) \in \mathbb{F}_{q^6}^3,$ with
\begin{equation}\label{solutiont=3_1}
\begin{split}
   &b^{q^s}=\frac{c(k^{q^{3s}+q^{5s}}-1)}{(h^{1+q^{2s}}-1)(k^{q^{s}+q^{3s}+q^{5s}}+1)^2},\\
&b^{q^{3s}}=\frac{-ck^{1+q^s+q^{5s}}(k^{q^{3s}}+1)}{(h+h^2)(k^{q^{s}+q^{3s}+q^{5s}}+1)^2},\\
&b^{q^{5s}}=\frac{ck^{q^s}(k^{{q^s}+q^{3s}}-1)}{ h^{1+q^{2s}}(h+h^{q^s})(k^{q^{s}+q^{3s}+q^{5s}}+1)^2}.
 \end{split}
\end{equation}

As in previous case we consider the following ratios
\begin{align}
&b^{q^{3s}-q^s}=-\frac{k^{1+q^{s}+q^{5s}}(k^{q^{3s}}+1)(h^{1+q^{2s}}-1)}{h(h+1)(k^{q^{3s}+q^{5s}}-1)}, \label{firstratio_1}\\
&b^{q^{5s}-q^{3s}}=-\frac{(k^{q^s+q^{3s}}-1)(h+1)}{h^{q^{2s}}(h+h^{q^s})k^{1+q^{5s}}(k^{q^{3s}}+1)}, \label{secondratio_1} \\ 
 & b^{q^{s}-q^{5s}}=\frac{h^{1+q^{2s}}(h+h^{q^s})(k^{q^{3s}+q^{5s}}-1)}{k^{q^s}(k^{q^s+q^{3s}}-1)(h^{1+q^{2s}}-1)} .\label{thirdratio_1}
\end{align}
Again, since the $q^{2s}$-power of Formulas \eqref{firstratio_1}, \eqref{secondratio_1} and \eqref{thirdratio_1} are Formulas \eqref{secondratio_1}, \eqref{thirdratio_1} and \eqref{firstratio_1}, respectively, we get the system below
\begin{equation}\label{rhosystem_1}
\begin{cases}
k^{q^s}h^{q^{2s}}=\left [\frac{(k^{q^s+q^{5s}}-1)(h+1)}{(h^{q^{4s}}+h^{q^{5s}})(k^{q^{3s}}+1)}\right]^{q^{2s}+1}\\
k^{q^{2s}}h^{q^{4s}+q^{2s}}=-\frac{(k^{q^s+q^{3s}}-1)(h^{q^{2s}}+1)}{(h+h^{q^s})(k^{q^{5s}}+1)}\\
k^{q^s+q^{5s}}h^{q^{s}}=-\frac{(k^{q^s+q^{5s}}-1)(h+1)}{(k^{q^{3s}}+1)(h^{q^{4s}}+h^{q^{5s}})}.
\end{cases}
\end{equation}
By setting $\xi=\frac{(k^{q^s+q^{5s}}-1)(h+1)}{(h^{q^{4s}}+h^{q^{5s}})(k^{q^{3s}}+1)}$, equations in (\ref{rhosystem_1}) becomes
\begin{equation}\label{rhosystem2_1}
\begin{cases}
k^{q^s}h^{q^{2s}}=\xi^{q^{2s}+1}\\
k^{q^{2s}}h^{q^{4s}+q^{2s}}=-\xi^{q^{2s}}\\
k^{q^s+q^{5s}}h^{q^s}=-\xi.
\end{cases}
\end{equation}
From the third equation, we have $\xi^{q^{3s}+1}=-1$, and so from the first and second equation we get, $h^{1+q^{2s}+q^{4s}}=\xi^{1+q^{2s}+q^{4s}}$. Now substituting the value of $k^{q^s+q^{5s}}$ and $k^{q^{3s}}$ from third and first equation to the expression of $\xi$ we get
\begin{equation*}
    \begin{split}
        &\xi=\frac{(\xi h^{q^{4s}}-1)(h+1)}{(-\xi^{q^{4s}+q^{2s}}h^{q^{s}}+1)(h^{q^{4s}}+h^{q^{5s}})},\\
       &{(h^{1+q^{2s}}+\xi)(h^{q^{4s}}+h^{q^{5s}})}=(\xi h^{q^{4s}}-1)(h+1)\\
        &\xi=h^{q^{2s}}.
    \end{split}
\end{equation*}
Using first equation of the system (\ref{rhosystem2_1}) we get $h=-k^{-1}$. Substituting this relation into (\ref{solutiont=3_1}), we get
\begin{equation}
    \begin{split}
      b^{q^{s}}&=\frac{c}{(k^{q^{s}+q^{3s}+q^{5s}}+1)^2}, \\
      b^{q^{3s}}&=-\frac{c k^{1+q^{s}+q^{5s}}}{k^{q^{3s}}(k^{q^{s}+q^{3s}+q^{5s}}+1)^2}  ,\\
      b^{q^{5s}}&=\frac{ck^{2q}}{k^{q^{3s}+q^{5s}}(k^{q^{s}+q^{3s}+q^{5s}}+1)^2} \,.
    \end{split}
\end{equation}
From the first two equations we easily get $c=\lambda k^{q^{3s}-q^{s}}$,where $\lambda\in \mathbb{F}_{q^2}$. This gives $b= \frac{\lambda^q k^{q^{2s}-1}}{(k^{1+q^{2s}+q^{4s}}+1)^2}$.
\end{proof}

Regarding the case when $t=4$, we prove the following result.

\begin{lemma}\label{lm:case_t_=4}
Assume $n=8$ and that $\mathcal{C}_{h,4,s}$ and \,$\mathcal{C}_{k,4,\ell}$, are equivalent. One gets the following: 
\begin{itemize}
\item[\textnormal{1.}] if $\ell \equiv s \pmod 8$, then $h^{\rho}=\pm k,$
\item[\textnormal{2.}] if $\ell \equiv -s \pmod 8$, then $h^{\rho}=\pm k^{-1},$
\item[\textnormal{3.}] if $\ell \equiv 3s \pmod 8$, then $h^{\rho}=\pm k,$
\item[\textnormal{4.}] if $\ell \equiv 5s \pmod 8$, then $h^{\rho}=\pm k^{-1},$
\end{itemize}
where $\rho \in \Aut(\F_{q^8})$. 
\end{lemma}
\begin{proof}

\noindent Assume that $\mathcal{C}_{h,4,s}=\langle x,\psi_{h,t,s}(x) \rangle_{\F_{q^8}}$ and\, $\mathcal{C}_{k,4,\ell}=\langle x,\psi_{k,\ell,s}(x) \rangle_{\F_{q^8}},$ are equivalent. By Corollary \ref{eq-Snq}, this may happen if and only if Formula \eqref{eq-standard} holds true for any pair of standard forms of the $q$-polynomials $\psi_{h,4,s}$ and $\psi_{k,4,\ell}$, respectively.
As mentioned before, the polynomials $\psi_{h,t,s}$ with $t \geq 4$ and even, is already expressed in standard form. Therefore, by Corollary \ref{eq-Snq}, in order to investigate the equivalence issue between two elements in this class, we only need to consider the two following equations:
\begin{equation}\label{eq:equiv_t=4}
		d\psi_{k,4,\ell}(x)=\psi^{\rho}_{h,4,s}(ax) \quad \textnormal{and} \quad \psi_{h,4,s}\left( b \psi_{k,4,\ell}^\rho(x)\right)=cx
\end{equation}
for $a, b, c, d \in \F^*_{q^8}$ and $\rho \in \Aut(\F_{q^8})$. Also here, without loss of generality, again by \cite[Remark 4.5]{neri_santonastaso_zullo}, we will suppose $\rho$ to be the identity.

\noindent Because of what asserted in the first part of the proof, we can proceed by dividing the remaining part of the argument into four cases.\\
\textbf{Case 1.} We assume $\ell \equiv s \pmod{8}$. Now, if $d\psi_{k,4,s}(x)=\psi_{h,4,s}(ax)$, then by comparing the coefficients, one easily gets $d=a^{q^s}$ with $a \in \F_{q^2}$ and $h=\pm k$.\\
If otherwise $\psi_{h,4,s}(b\psi_{k,4,\ell}(x))=cx$, by expanding left side of this equation, one get

	\begin{align*}
		&\left(b^{q^s} + b^{q^{3s}}k^{q^{3s}-q^{2s}} +b^{q^{5s}}h^{1-q^{5s}}k^{q^{5s}-q^{2s}}+ b^{q^{7s}}h^{1-q^{7s}}\right)x^{q^{2s}} \\
		+&\left(b^{q^s}+b^{q^{3s}}-b^{q^{5s}}k^{q^{5s}-q^{4s}}h^{1-q^{5s}} -b^{q^{7s}}h^{1-q^{7s}}k^{q^{7s}-q^{4s}}\right)x^{q^{4s}}\\
		+&\left( -k^{q^s-q^{6s}}b^{q^s}+b^{q^{3s}}-h^{1-q^{5s}}b^{q^{5s}}+h^{1-q^{7s}}k^{q^{7s}-q^{6s}}b^{q^{7s}} \right)x^{q^{6s}}\\
		+&\left(b^{q^s}k^{q^s-1} -b^{q^{3s}}k^{q^{3s}-1}-h^{1-q^{5s}}b^{q^{5s}}+h^{1-q^{7s}}b^{q^{7s}}\right)x=cx.
	\end{align*}
	
By equating the coefficients of the terms $x, x^{q^{2s}}, x^{q^{4s}}$, and $x^{q^{6s}}$ on the left and right side, and taking into account that $h^{q^{4s}+1}=k^{q^{4\ell}+1}=-1$, we obtain the following linear system in $b^{q^s},b^{q^{3s}},b^{q^{5s}}$, and $b^{q^{7s}}$:

\begin{equation}\label{linear-system}
    \begin{pmatrix}
    1 & k^{q^{3s}-q^{2s}} & h^{1-q^{5s}}k^{q^{5s}-q^{2s}}& h^{1-q^{7s}} \\
    1 & 1 & -k^{q^{5s}-q^{4s}}h^{1-q^{5s}} & -h^{1-q^{7s}}k^{q^{7s}-q^{4s}} \\
    -k^{q^s-q^{6s}} & 1 &  -h^{1-q^{5s}} & h^{1-q^{7s}}k^{q^{7s}-q^{6s}}\\
    k^{q^s-1} & -k^{q^{3s}-1} & -h^{1-q^{5s}} & h^{1-q^{7s}}
    \end{pmatrix}
    \begin{pmatrix}
    b^{q^s} \\
    b^{q^{3s}} \\
    b^{q^{5s}} \\
    b^{q^{7s}}
    \end{pmatrix}=
    \begin{pmatrix}
    0 \\
    0 \\
    0 \\
    c \\
    \end{pmatrix}.
\end{equation}
Let us denote by $A_{h,k,s} \in \F_{q^8}^{4 \times 4}$ the matrix associated with it. First, we show that its determinant is not zero. In fact,  by recalling that $k^{1+q^{4s}}=-1$ and that $(s,8)=1,$ we get

\[\det(A_{h,k,s})=-\mathrm{Tr}_{q^8/q}(k^{q^s+1}+k^{q^{3s}+1})h^{2 - q^{5s} - q^{7s}}.\]
\noindent
If $\det(A_{h,k,s})=0$, since
\begin{equation*}
\mathrm{Tr}_{q^8/q}(k^{q^s+1}+k^{q^{3s}+1})=\mathrm{Tr}_{q^8/q^{2}}(k)^{q^s+1},
\end{equation*}
we would get that
\begin{equation*}
    k+k^{q^{2s}}+k^{q^{4s}}+k^{q^{6s}}=0
\end{equation*}
or equivalently that 
\begin{equation*}
    k+k^{q^{2s}}=\frac{1}{k}+\frac{1}{k^{q^{2s}}}.
\end{equation*}
Then
\begin{equation*}
    k^{q^{2s}+1}=1,
\end{equation*}
which, by \cite[Proposition 3.2]{longobardi_marino_trombetti_zhou}, can not be the case. Since $c \not = 0$, then the unique solution of System  \eqref{linear-system} is ${\bf b}=\big (b^{q^s},b^{q^{3s}},b^{q^{5s}},b^{q^{7s}}\big ) \in \mathbb{F}_{q^8}^4,$ with
\begin{equation}\label{b-expressions_case_1}
\begin{split}
    &b^{q^s}=\frac{c}{T}(k^{q^{5s}+q^{7s}}-1)(k^{q^{6s}-q^{4s}}+1)\\
       &b^{q^{3s}}=-\frac{c}{T}(k^{q^{5s}+q^{7s}}-1)(k^{q^{2s}-q^{4s}}+1)\\
       &b^{q^{5s}}=-\frac{c}{h^{1-q^{5s}}T}(k^{q^s}+k^{q^{7s}})(k^{q^{2s}}+k)\\
       &b^{q^{7s}}=\frac{c}{h^{1-q^{7s}}T}(k^{q^{3s}}+k^{q^{5s}})(k^{q^{6s}}+k),
       \end{split}
\end{equation}
where $T=\mathrm{Tr}_{q^8/q}(k^{1+q^s}+k^{1+q^{3s}})$. By computing the ratio between second and first component of ${\bf b}$ we have
\begin{equation}\label{bq}
    (b^{q^s})^{q^{2s}-1}=-\frac{k^{q^{2s}-q^{4s}}+1}{k^{q^{6s}-q^{4s}}+1}=-\frac{k^{q^{2s}}+k^{q^{4s}}}{k^{q^{4s}}+k^{q^{6s}}}=-\frac1{(k^{q^{2s}}+k^{q^{4s}})^{q^{2s}-1}}.
\end{equation}
In the same way, computing the ratio between fourth  and third equation in (\ref{b-expressions_case_1}) we have
\begin{equation}\label{bq5}
\begin{split}
    (b^{q^{5s}})^{q^{2s}-1}&=-h^{q^{7s}-q^{5s}}\frac{(k^{q^{6s}}+k)(k^{q^{3s}}+k^{q^{5s}})}{(k^{q^{2s}}+k)(k^{q^s}+k^{q^{7s}})}\\
    &=-\biggl(\frac{h^{q^{5s}}k^{q^{s}}}{k^{q^{6s}}+k}\biggr)^{q^{2s}-1}\frac{(k^{q^{5s}-q^{3s}}+1)}{(k^{q^{7s}-q^s}+1)}=-\biggl(\frac{h^{q^{5s}}k^{q^{s}}}{k^{q^{6s}}+k}\biggr)^{q^{2s}-1}.
    \end{split}
\end{equation}
Let $\omega$ be an element in $\F_{q^8}$ such that $\omega^{q^{2s}-1}=-1$; then, from the last two Expressions (\ref{bq}) and (\ref{bq5}) we get
\begin{equation}\label{solutions}
b^{q^s}=\frac{\lambda \omega}{k^{q^{2s}}+k^{q^{4s}}}  \,\,\,\text{ and } \,\,\, b^{q^{5s}}=\frac{h^{q^{5s}}k^{q^{s}}\mu \omega}{k^{q^{6s}}+k},
\end{equation} respectively,
where $\lambda,\mu \in \F_{q^2}$.
Raising the two sides of the first equation in (\ref{solutions}) to the $q^{4s}$ power, and taking into account the second equation, we have
that $h^{q^{5s}}k^{q^{s}}$ must also belong to $\F_{q^2}$. In other words we get

\begin{equation}\label{eq:l=s_cond_on_h_and_k}
h^{q^{5s}}=\nu k^{-q^{s}}, \end{equation} for some $\nu \in \F^*_{q^2}$.\\
Now, taking into account Equations  \eqref{solutions}, recalling that $\omega^{q^{2s}}=-\omega$ and that $\lambda \in \F_{q^2}$, the first equation of the Linear System \eqref{linear-system} becomes
\begin{equation}
    \frac1{k^{q^{2s}}+k^{q^{4s}}} -\frac{k^{q^{3s}-q^{2s}}}{k^{q^{4s}}+k^{q^{6s}}}+\frac{h^{1-q^{5s}}k^{q^{5s}-q^{2s}}}{k^{q^{6s}}+k}-\frac{h^{1-q^{7s}}}{k+k^{q^{2s}}}=0.
\end{equation}

Plugging (\ref{eq:l=s_cond_on_h_and_k}), latter equation reads
\begin{equation}
    \frac1{k^{q^{2s}}+k^{q^{4s}}} -\frac{k^{q^{3s}-q^{2s}}}{k^{q^{4s}}+k^{q^{6s}}}=\nu^{q^s-1}\biggl(\frac{k^{-q^{4s}-q^{2s}}}{k^{q^{6s}}+k}+\frac{k^{q^{3s}-q^{4s}}}{k+k^{q^{2s}}}\biggr),
\end{equation}

whence
\begin{equation}\label{solutions_2}
    \frac1{k^{q^{2s}}+k^{q^{4s}}} -\frac{k^{q^{3s}}}{k^{q^{4s}+q^{2s}}-1}=-\nu^{q^s-1}\biggl(\frac1{k^{q^{2s}}+k^{q^{4s}}} -\frac{k^{q^{3s}}}{k^{q^{4s}+q^{2s}}-1}\biggr).
\end{equation}

We observe here that the term in the left side of Equation (\ref{solutions_2}), can not be equal to zero. Indeed, if it was  \[\frac1{k^{q^{2s}}+k^{q^{4s}}} -\frac{k^{q^{3s}}}{k^{q^{4s}+q^{2s}}-1}=0; \]
then, raising this equation to the $(q^{4s}+1)$-th power, and taking into account that $k^{q^{4s}+1}=-1$, we would get 
\begin{equation}\label{itcanbezero}
(k^{q^{2s}}-k^{q^{6s}})(k-k^{q^{4s}})=0,
\end{equation}
and so $k \in \F_{q^4}$. In such a case, the condition $k^{1+q^{4s}}=-1$ implies that $k \in \F_{q^2}$. Then, by \eqref{bq}, $b^{q^s}=\lambda \omega$ where  $\omega^{q^{2s}-1}=-1$. Using the expression of $b^{q^{5s}}$ found in \eqref{bq5}, we get that $h \in \F_{q^2}$ as well. Then, first and third equation of Linear System \eqref{linear-system}, can be rewritten in the following form
\begin{equation}
\begin{cases}
(1-k^{{q^s}-1})(1-h^{1-q^s})=0\\
(1+k^{{q^s}-1})(1+h^{1-q^{s}})=0.
\end{cases}
\end{equation}
These two equations are verified together only when either $k \in \F_q$ and $h \in \F_{q^2}\setminus \F_q$ with $h^{q^s-1}=-1$ or $h \in \F_q$ and $k \in \F_{q^2}\setminus \F_q$ with $k^{q^s-1}=-1$. However, both when $q \equiv 1 \pmod 4$ and $q \equiv 3 \pmod 4$, it is straightforward to see that these cases cannot occur.
Then, the expression  in \eqref{itcanbezero} is not zero and hence, $\nu^{q^s-1}=-1$. On the other hand, since $\nu \in \mathbb{F}_{q^2},$ we also have $k^{1+q^{4s}}=\frac{\nu^2}{h^{1+q^{4s}}}$, whence $\nu^2=1$, a contradiction. \\

\noindent \textbf{Case 2.} We assume $\ell \equiv -s \pmod 8$. If 
$d\psi_{k,4,-s}(x)=\psi_{h,4,s}(ax),$
by equating the coefficients we obtain
\begin{equation}\label{a-d-system}
\begin{cases}
dk^{1-q^s}=a^{q^s}\\
-dk^{1-q^{3s}}=a^{q^{3s}}\\
d=-a^{q^{5s}}h^{1-q^{5s}}\\
d=a^{q^{7s}}h^{1-q^{7s}}.
\end{cases}   
\end{equation}
From the first and second equation, we get
\begin{equation}
    (a^{q^s})^{q^{2s}-1}=-\left (\frac1{k^{q^s}} \right )^{q^{2s}-1}
\end{equation}
and hence
\begin{equation}\label{aqs}
  a^{q^s}=\frac{\lambda \omega}{k^{q^s}}  
\end{equation}
with $\omega^{{q^{2s}-1}}=-1$ and $\lambda \in \F_{q^2}^*$; whence, $d=\frac{\omega \lambda}{k}$. On the other hand, from the third and the last equation, we get
\begin{equation}
    (a^{q^{5s}})^{q^{2s}-1}=-(h^{q^{5s}})^{q^{2s}-1}
\end{equation}
and hence 
\begin{equation}\label{aq5s}
    a^{q^{5s}}=\mu \omega h^{q^{5s}}
\end{equation}
with $\mu \in \F_{q^2}^*$ and so $d=-\omega \mu h$. Raising to the $q^{4s}$-power the Formula \eqref{aqs} and getting it equals to \eqref{aq5s} , we obtain
\begin{equation}\label{nu}
   h^{q^{5s}}k^{q^{5s}}=\frac{\lambda}{\mu}=\nu 
\end{equation}
In the same way, by the expressions of $d$, we get
\begin{equation}\label{-nu}
    hk=-\nu.
\end{equation}
By assumption on $h$ and $k$, and since $\nu \in \F_{q^2},$ we have that  $\nu^2=(hk)(hk)^{q^{4s}}=1$; then, $\nu= \pm 1$. Hence, by \eqref{nu} and \eqref{-nu}

\[
\pm 1=(hk)^{q^{5s}}=hk= \mp 1,
\]
a contradiction.
Then, this case cannot occur.\\
If otherwise $\psi_{h,s}(b\psi_{k,-s}(x))=cx$ then, expanding left side of this equation, we get 

\begin{align*}
& \left (b^{q^s}  -
 b^{q^{3 s}} k^{ q^{3 s}-q^{6s}}-
  b^{q^{5 s}} h^{1 -q^{5s}} k^{q^{5 s} - q^{6 s}} +
 b^{q^{7 s}} h^{1 - q^{7 s}}     
\right) x^{q^{6 s}}\\
&+\left (b^{q^s}  + b^{q^{3 s}}  + 
 b^{q^{5 s}} h^{1 - q^{5s}} k^{q^{5 s}-1}  + 
 b^{q^{7 s}} h^{1 - q^{7 s}} k^{q^{7 s}-1}\right )x\,
 \\
 +& \left ( b^{q^s} k^{q^s - q^{2 s}}+b^{q^{3 s}}  -
 b^{q^{5 s}} h^{1 - q^{5s}}   - b^{q^{7 s}} h^{1 - q^{7 s}} k^{q^{7 s}-q^{2s}} \right) x^{q^{2 s}}  \\
 + & \left(- b^{q^s} k^{ q^s-q^{4s}}+ 
 b^{q^{3 s}} k^{q^{3 s} - q^{4 s}}-b^{q^{5 s}} h^{1 - q^{5s}}  + 
 b^{q^{7 s}} h^{1 - q^{7 s}}     \right) x^{q^{4 s}} =cx.
	\end{align*}

Comparing the coefficients of $x, x^{q^{2s}}, x^{q^{4s}},$ and $x^{q^{6s}}$ on the left and right side, and taking into account that $h^{q^{4s}+1}=k^{q^{4\ell}+1}=-1$, we obtain the following linear system in $b^{q^s},b^{q^{3s}},b^{q^{5s}},b^{q^{7s}}$:
\begin{equation}\label{l=-1}
\begin{pmatrix}
    1 & -k^{q^{3s}-q^{6s}} & -h^{1-q^{5s}}k^{q^{5s}-q^{6s}}& h^{1-q^{7s}} \\
    1 & 1 & k^{q^{5s}-1}h^{1-q^{5s}} & h^{1-q^{7s}}k^{q^{7s}-1} \\
    k^{q^s-q^{2s}} & 1 &  -h^{1-q^{5s}} & -h^{1-q^{7s}}k^{q^{7s}-q^{2s}}\\
    -k^{q^s-q^{4s}} & k^{q^{3s}-q^{4s}} & -h^{1-q^{5s}} & h^{1-q^{7s}}
    \end{pmatrix}
    \begin{pmatrix}
    b^{q^s} \\
    b^{q^{3s}} \\
    b^{q^{5s}} \\
    b^{q^{7s}}
    \end{pmatrix}=
    \begin{pmatrix}
    0 \\
    c \\
    0 \\
    0 \\
    \end{pmatrix}.
\end{equation}

We denote by $A_{h,k,-s}$ the coefficient matrix associated with this system. It is not difficult to see that ${\rm det} (A_{h,k,-s})={\rm det} (A_{h,k,s})$, hence it is not zero and there is a unique solution ${\bf b}=\big (b^{q^s},b^{q^{3s}},b^{q^{5s}},b^{q^{7s}}\big ) \in \mathbb{F}_{q^8}^4$ of System \eqref{l=-1} with components:
\begin{equation}\label{b-expressionsl=-1}
\begin{split}
    &b^{q^s}=\frac{c}{T}(k^{q^{2s}}+k)(k^{q^{3s}}+k^{q^{5s}})\\
    &b^{q^{3s}}=\frac{c}{T}(k^{q^{6s}}+k)(k^{q^s}+k^{q^{7s}})\\
    &b^{q^{5s}}=-\frac{c}{h^{1-q^{5s}}T}(k^{q^{6s}+1}-1)(k^{q^{3s}+q^{s}}-1)\\
    &b^{q^{7s}}=-\frac{c}{h^{1-q^{7s}}T}(k^{q^{2s}+1}-1)(k^{q^{3s}+q^{s}}-1)
 \end{split}
\end{equation}
Computing the ratio between second and first equation in (\ref{b-expressionsl=-1}), we have
\begin{equation*}
\begin{split}
    (b^{q^{s}})^{q^{2s}-1}&=\frac{(k^{q^{6s}}+k)(k^{q^s}+k^{q^{7s}})}{(k^{q^{2s}}+k)(k^{q^{3s}}+k^{q^{5s}})}=\frac{(k^{q^{6s}}+k)k^{q^s}(1+k^{q^{7s}-q^s})}{(k^{q^{2s}}+k)k^{q^{3s}}(1+k^{q^{5s}-q^{3s}})}=\biggl(\frac{k^{q^{5s}}}{k^{q^{6s}}+k}\biggr)^{q^{2s}-1}.
    \end{split}
\end{equation*}
Similarly, dividing third equation by the fourth, we get
\begin{equation*}
    \begin{split}
        (b^{q^{5s}})^{q^{2s}-1}&=h^{q^{7s}-q^{5s}}\frac{k^{q^{2s}}+k^{q^{4s}}}{k^{q^{6s}}+k^{q^{4s}}}=\biggl(\frac{h^{q^{5s}}}{k^{q^{4s}}+k^{q^{2s}}}\biggr)^{q^{2s}-1}.
    \end{split}
\end{equation*}
The last two expressions together imply that
\begin{equation}\label{solutions2}
\begin{split}
      b^{q^s}&= \lambda \frac{k^{q^{5s}}}{k^{q^{6s}}+k}\\
      b^{q^{5s}}&=\mu\frac{h^{q^{5s}}}{k^{q^{4s}}+k^{q^{2s}}}
      \end{split}
\end{equation}
with $\lambda,\mu \in \F_{q^2}$. 
Now, from the two equations in (\ref{solutions2}) we easily get $\frac{h^{q^{5s}}}{k^{q^{s}}}=\frac{\lambda}{\mu}= \nu  $, where $\nu \in \F^*_{q^2}$. As before,
\begin{equation*}
    \nu^2=\frac{h^{q^{5s}}}{k^{q^{s}}}\cdot \frac{h^{q^{s}}}{k^{q^{5s}}}=1.
\end{equation*}
Then $\nu=  \pm 1$, which in turn implies that $h = \pm k^{-1}$, and the following values for $b$ and $c$:
\begin{equation*}
b= \frac{\lambda^{q^s}k^{q^{4s}}}{k^{q^{7s}}+k^{q^{5s}}}, \quad \quad c=\lambda \left(  k^{q^{7 s}-1} +  k^{q^{5s}-1} \right).
\end{equation*}
This concludes this case.\\ 

\noindent \textbf{Case 3.} We assume now $\ell \equiv 3s \pmod 8 $. If $d\psi_{k,4,3s}(x)=\psi_{h,4,s}(ax)$, by expanding and comparing the coefficients, we get the following set of conditions
 \begin{equation}
 \begin{cases}
     d=a^{q^s}\\
     d=a^{q^{3s}}\\
     dk^{1-q^{5s}}=-a^{q^{5s}}h^{1-q^{5s}}\\
     -dk^{1-q^{7s}}=a^{q^{7s}}h^{1-q^{7s}}.
 \end{cases}
 \end{equation}
 By the first and the second, $d=a^{q^s}$ with $a \in \F_{q^2}$ and plugging this into the third equation, we obtain 
 \begin{equation}
     \left (\frac{h}{k}\right )^{q^{5s}-1} =-1.
 \end{equation}
Then, $h=\omega k$, where $\omega \in \F_{q^2}$ such that $\omega^{{q^s}-1}=-1$. Since on the other hand we also have $\omega^2=(h/k)^{1+q^{4s}}=1$, this leads to a contradiction; hence, this case cannot occur.\\
Assume that  $\psi_{h,4,s}(b\psi_{k,4,3s}(x))=cx$. Expanding this equation we get 
	\begin{align*}
	&\left (b^{q^s} -k^{ q^{3 s}-q^{2s}}b^{q^{3 s}}  -
 b^{q^{5 s}} h^{1 - q^{5s}} k^{q^{5 s} - q^{2 s}}  + 
 b^{q^{7 s}} h^{1 - q^{7 s}}\right )x^{q^{2s}}\,
 \\
 +& \left (b^{q^s} +b^{q^{3 s}}  + 
 b^{q^{5 s}} h^{1 - q^{5s}} k^{q^{5s} -q^{4s}}  +  b^{q^{7 s}} h^{1 - q^{7 s}} k^{q^{7 s} - q^{4 s}} \right) x^{q^{4 s}}  \\
 + & \left(b^{q^s} k^{q^s - q^{6s}}  +b^{q^{3 s}} - 
 b^{q^{5 s}} h^{1 - q^{ 5s}}-  
 b^{q^{7 s}} k^{q^{7 s} - q^{6 s}}h^{1 - q^{ 7s}} \right) x^{q^{6 s}}
 \\
 + & \left (-b^{q^s}k^{q^{ s} - 1}  + 
 b^{q^{3 s}} k^{ q^{3 s}-1}  -
 b^{q^{5 s}}h^{1 - q^{5s}} + 
 b^{q^{7 s}} h^{1 - q^{7s}} \right) x=cx.
	\end{align*}
Now, by comparing the coefficients of $x, x^{q^{2s}}, x^{q^{4s}},$ and $x^{q^{6s}}$ on the left and right side, the equation above gives the following linear system in $b^{q^s},b^{q^{3s}},b^{q^{5s}},b^{q^{7s}}$:
\begin{equation}\label{linear-system3}
    \begin{pmatrix}
    1 & -k^{q^{3s}-q^{2s}} & -h^{1-q^{5s}}k^{q^{5s}-q^{2s}}& h^{1-q^{7s}} \\
    1 & 1 & k^{q^{5s}-q^{4s}}h^{1-q^{5s}} & h^{1-q^{7s}}k^{q^{7s}-q^{4s}} \\
    k^{q^s-q^{6s}} & 1 &  -h^{1-q^{5s}} & -h^{1-q^{7s}}k^{q^{7s}-q^{6s}}\\
    -k^{q^s-1} & k^{q^{3s}-1} & -h^{1-q^{5s}} & h^{1-q^{7s}}
    \end{pmatrix}
    \begin{pmatrix}
    b^{q^s} \\
    b^{q^{3s}} \\
    b^{q^{5s}} \\
    b^{q^{7s}}
    \end{pmatrix}=
    \begin{pmatrix}
    0 \\
    0 \\
    0 \\
    c \\
    \end{pmatrix}.
\end{equation}
The determinant of the matrix $A_{h,k,3s}$ associated with linear system (\ref{linear-system3}) is equal to
\begin{equation*}
\mathrm{Tr}_{q^8/q}(k^{q
    ^s+1}+k^{q^{3s}+1})h^{2 - q^{5s} - q^{7s}}
\end{equation*}
and hence it is different from zero.

Since $c \not = 0$ then, the unique solution of System  \eqref{linear-system3} has components
\begin{equation}\label{b-expressions}
\begin{split}
    &b^{q^s}=-\frac{c}{T}(k^{q^{5s}+q^{7s}}-1)(k^{q^{6s}-q^{4s}}+1)\\
       &b^{q^{3s}}=\frac{c}{T}(k^{q^{5s}+q^{7s}}-1)(k^{q^{2s}-q^{4s}}+1)\\
       &b^{q^{5s}}=-\frac{c}{h^{1-q^{5s}}T}(k^{q^s}+k^{q^{7s}})(k^{q^{2s}}+k)\\
       &b^{q^{7s}}=\frac{c}{h^{1-q^{7s}}T}(k^{q^{3s}}+k^{q^{5s}})(k^{q^{6s}}+k),
       \end{split}
\end{equation}
which is the same  solution obtained for System (\ref{linear-system}) in the case where $\ell\equiv s\pmod{8}$. Then arguing as in that case, we get $h^{q^{5s}}=\nu k^{-q^{s}}$ and $b^{q^s}=\frac{\lambda \omega}{k^{q^{2s}}+k^{q^{4s}}}$ for some $\nu, \lambda \in \F^*_{q^2}$ and  $\omega\in\F_{q^8}$ such that $\omega^{q^{2s}-1}=-1$. Since $1=(h^{q^{5s}}k^{q^s})^{1+q^{4s}}=\nu^2$, then $\nu=\pm1$  and so $h=\pm k$,
\begin{equation*}
    b=-\frac{\lambda^{q^s}\omega^{q^s}}{k^{q^{3s}}+k^{q^s}} \,\textnormal{ and } \,c=\lambda \omega (k^{q^{s}} + k^{q^{3s}} ).
    \end{equation*}
\noindent \textbf{Case 4.} Finally, we assume $\ell \equiv 5s \pmod 8 $. If 
$d\psi_{k,4,5s}(x)=\psi_{h,4,s}(ax)$
by equating the coefficients we obtain the following conditions:
\begin{equation}\label{a-d-system1}
\begin{cases}
-dk^{1-q^s}=a^{q^s}\\
dk^{1-q^{3s}}=a^{q^{3s}}\\
d=-a^{q^{5s}}h^{1-q^{5s}}\\
d=a^{q^{7s}}h^{1-q^{7s}}.
\end{cases}   
\end{equation}
Proceeding as we did in the case where $s\equiv -\ell\pmod{8}$, we get $h=\pm k^{-1}$, $a=- \lambda^{q^s}\omega^{q^s}k^{-1}$ and $d=-\frac{\lambda \omega}{k}$ where $\lambda \in \F_{q^2}$ and $\omega\in \F_{q^8}$ such that $\omega^{{q^{2s}}}=-\omega$.

If instead we assume $\psi_{h,4,s}(b\psi_{k,4,5s}(x))=cx$, then, always by expanding this equation, we get

	\begin{align*}
 & \left (b^{q^s}  + b^{q^{3 s}} k^{q^{3 s} - q^{6 s}}  + b^{q^{5 s}} h^{1 - q^{5s}} k^{q^{5 s} + q^{2 s}}+
 b^{q^{7 s}} h^{1 - q^{7 s}}\right) x^{q^{6 s}}\\
 +&\left (b^{q^s}  + b^{q^{3 s}}  - 
 b^{q^{5 s}} h^{1 - q^{5s}} k^{ q^{5 s}-1}  - 
 b^{q^{7 s}} h^{1 - q^{7 s}} k^{q^{7 s}-1}\right )x\\
	&+ \left (-b^{q^s} k^{q^s - q^{2 s}}  +b^{q^{3 s}}  - 
 b^{q^{5 s}} h^{1 - q^{5s}}  +  b^{q^{7 s}} h^{1 - q^{7 s}} k^{q^{7 s} - q^{2 s}} \right) x^{q^{2 s}}  \\
 + & \left(b^{q^s} k^{q^s-q^{4s}}  -  
 b^{q^{3 s}} k^{q^{3 s} -q^{4s}}-b^{q^{5 s}} h^{1 - q^{5s}}  + 
 b^{q^{7 s}} h^{1 - q^{7 s}}   \right) x^{q^{4 s}}=cx.
	\end{align*}
	
Now, comparing the coefficients of $x, x^{q^{2s}}, x^{q^{4s}},$ and $x^{q^{6s}}$ on the left and right side of this equation, and taking into account that $h^{q^{4s}+1}=k^{q^{4\ell}+1}=-1$, we obtain the following linear system in $b^{q^s},b^{q^{3s}},b^{q^{5s}},b^{q^{7s}}$:
\begin{equation}\label{l=5}
\begin{pmatrix}
 1 & k^{q^{3s}-q^{6s}} & h^{1-q^{5s}}k^{q^{5s}-q^{6s}}& h^{1-q^{7s}} \\
 1 & 1 & -k^{q^{5s}-1}h^{1-q^{5s}} & -h^{1-q^{7s}}k^{q^{7s}-1} \\
     -k^{q^s-q^{2s}} & 1 &  -h^{1-q^{5s}} & h^{1-q^{7s}}k^{q^{7s}-q^{2s}}\\
      k^{q^s-q^{4s}} & -k^{q^{3s}-q^{4s}} & -h^{1-q^{5s}} & h^{1-q^{7s}}\\

    \end{pmatrix}
    \begin{pmatrix}
    b^{q^s} \\
    b^{q^{3s}} \\
    b^{q^{5s}} \\
    b^{q^{7s}}
    \end{pmatrix}=
    \begin{pmatrix}
    0 \\
    c\\
    0 \\
    0 \\
    \end{pmatrix}
\end{equation}

As in the previous cases the determinant of the matrix $A_{h,k,5s}$ associated with relevant linear system (\ref{l=5}) is equal to $-\mathrm{Tr}_{q^8/q}(k^{q
    ^s+1}+k^{q^{3s}+1})h^{2 - q^{5s} - q^{7s}}$, and hence it is different from zero. 
Since $c \not = 0$ then, the unique solution of System  \eqref{l=5} is
\begin{equation}\label{b-expressions5}
\begin{split}
    &b^{q^s}=\frac{c}{T}(k^{q^{2s}}+k)(k^{q^{3s}}+k^{q^{5s}})\\
    &b^{q^{3s}}=\frac{c}{T}(k^{q^{6s}}+k)(k^{q^s}+k^{q^{7s}})\\
    &b^{q^{5s}}=-\frac{c}{h^{1-q^{5s}}T}(k^{q^{6s}+1}-1)(k^{q^{3s}+q^{s}}-1)\\
    &b^{q^{7s}}=-\frac{c}{h^{1-q^{7s}}T}(k^{q^{2s}+1}-1)(k^{q^{3s}+q^{s}}-1).
 \end{split}
\end{equation}
These are the same values as found in the case where $s\equiv -\ell\pmod{8}$ (see Case 2, Equations (\ref{b-expressionsl=-1})). Hence, arguing as in that case, we get $b^{q^s}= \lambda \frac{k^{q^{5s}}}{k^{q^{6s}}+k}$ and

\begin{equation}\label{eq:l=5s_cond_on_h_and_k}
h^{q^{5s}}=\nu{k^{q^{s}}}, \end{equation} where $\nu \in \F^*_{q^2}.$
 Substituting the values of $b^{q^s},$ $b^{q^{3s}},$ $
b^{q^{5s}}$ and $b^{q^{7s}}$ expressed in (\ref{b-expressions5}), in the first equation of linear System (\ref{l=5}) we get,
\begin{equation}
    \frac{k^{q^{5s}}}{k^{q^{6s}}+k}-\frac{k^{-q^{6s}}}{k+k^{q^{2s}}}+\frac{h^{1-q^{5s}}k^{q^{5s}-q^{6s}}k^{q^s}}{k^{q^{2s}}+k^{q^{4s}}}+\frac{h^{1-q^{7s}}k^{q^{3s}}}{k^{q^{6s}}+k^{q^{4s}}}=0.
   \end{equation}
   
   Then, taking into account Equation (\ref{eq:l=5s_cond_on_h_and_k}), we get
   \begin{equation}\label{nu_ell_cong_5s}
       \begin{split}
    &\frac{k^{q^{5s}}}{k^{q^{6s}}+k}-\frac{1}{k^{1+q^{6s}}-1}=\nu^{q^s-1}\biggl(\frac{-k^{q^{5s}-q^{6s}+q^{4s}}}{k^{q^{2s}}+k^{q^{4s}}}-\frac{k^{q^{4s}}}{k^{q^{6s}}+k^{q^{4s}}}\biggr)\\
    &\frac{k^{q^{5s}}}{k^{q^{6s}}+k}-\frac{1}{k^{1+q^{6s}}-1}=\nu^{q^s-1}\biggl(-\frac{k^{q^{5s}}}{-k^{-q^{4s}}+k^{q^{6s}}}-\frac{1}{-k^{q^{6s}+1}+1}\biggr)\\
    &\frac{k^{q^{5s}}}{k^{q^{6s}}+k}-\frac{1}{k^{1+q^{6s}}-1}=-\nu^{q^s-1}\biggl(\frac{k^{q^{5s}}}{k^{q^{6s}}+k}-\frac{1}{k^{1+q^{6s}}-1}\biggr),
       \end{split}
   \end{equation}
The left side of Equations in \eqref{nu_ell_cong_5s} cannot be zero otherwise, as it happens in the case where $\ell \equiv s \pmod 8$, a contradiction follows.
Then, $\nu^{q^s-1}=-1$ and taking the  $(1+q^{4s})$-th power of Equation (\ref{eq:l=5s_cond_on_h_and_k}), we get $\nu^2=1$. Again a contradiction. Hence, this case cannot occur and this concludes the proof. 
\end{proof}

We are now in the position to prove the following main theorem.

\begin{theorem}\label{main_result}
Let  $n=2t,$ $t \in \{3,4\}$, $h, k \in \F_{q^{n}}$ satisfying  $\N_{q^{n}/q^t}(h)=  \N_{q^n/q^t}(k)= -1$. Let $s, \ell \in \mathbb{N}$ such that $(n,s)=(n,\ell)=1$. Finally, let ${\cal C}_{h,t,s}$ and ${\cal C}_{k,t,\ell}$ be two elements in Class (\ref{neri_santonastaso_zullo}). Then: 
\begin{itemize}
    \item[(i)] ${\cal C}_{h,3,s}$ and ${\cal C} _{k,3,\ell}$ are equivalent if only if either $s \equiv \ell \pmod n$ and $h^\rho = \pm k$, or $s \equiv -\ell \pmod n$ and $h^\rho = \pm k^{-1}$ where $\rho \in \Aut(\F_{q^n})$.
\item[(ii)] ${\cal C}_{h,4,s}$ and ${\cal C} _{k,4,\ell},$ are equivalent if only if either $s \equiv \ell \pmod n $ and $h^\rho = \pm k$, or $s \equiv -\ell$ and $h^\rho= \pm k^{-1},$ or $s \equiv 3\ell \pmod n $ and $h^\rho = \pm k$, or $s \equiv 5\ell \pmod n$ and $h^\rho = \pm k^{-1}$ where $\rho \in \Aut(\F_{q^n})$.
\end{itemize}
\end{theorem}

\begin{proof} It is not difficult to show that the necessity condition in the statement of this theorem holds true. Regarding the sufficiency, we first note that since $(s,n)=(\ell,n)=1$, there exists an integer $ -1 \leq r \leq n-2$ such that $\ell \equiv s r \pmod n$. Moreover, it is straightforward to see that $r$ and $n$ must be coprime. Then, for $t=3$, $r  \not \in \{2,3,4\}$ and for $t=4$, $r \not \in \{2,4,6\}$.
Hence, one derives that if $t\in \{3,4\}$ and $r \neq  \pm1 $ or $r \neq t \pm1,$ then 
$\mathcal{C}_{h,t,s}$ and\, $\mathcal{C}_{k,t,{\ell}}$ can not be equivalent. Assume that $\mathcal{C}_{h,t,s}=\langle x,\psi_{h,t,s}(x) \rangle_{\F_{q^n}}$ and\, $\mathcal{C}_{k,t,\ell}=\langle x,\psi_{k,\ell,s}(x) \rangle_{\F_{q^n}},$ are equivalent. As a consequence of what stated above, if $t=3$ this may happen only if $r\in\{-1,1\}$; while if $t=4$ only if $r\in\{-1,1,3,5\}$. Hence, the result follows by Lemmas \ref{lm:case_t_=3} and \ref{lm:case_t_=4}.
\end{proof}

\medskip
\noindent We conclude this section pointing out that as a consequence of Theorem \ref{main_result} above, by using same argument as developed in \cite[Theorem 4.10]{neri_santonastaso_zullo}, it is easy to determine the exact number of the equivalence classes of codes in the family $\cC_{h,t,s}$ (or equivalently the number of inequivalent codes in it), also for $t \in \{3,4\}$. As in \cite[Theorem 4.12]{neri_santonastaso_zullo}, a lower bound for this number in the case $t \in \{3,4\}$, turns out to be  $\left \lfloor{\frac{\varphi(t)(q^t+ 1)}
{8rt}}  \right \rfloor$,
where $q=p^r$ and $\varphi$ is the Euler's totient function.

\section{Equivalence issue of linear sets}

In this section, we will denote by $L_{2,s}$ and $L_{2,s,\eta}$ and $L_{n,s,\delta}$ the linear sets associated with the codes ${\cal G}_{2,s}$, ${\cal H}_{2,s,\eta}$ and ${\cal K}_{n,s,\delta}$, respectively. Also, we will denote by $L_{h,t,s}$ the linear set of $\mathrm{PG}(1,q^{2t})$ associated with the code ${\cal C}_{h,t,s}$, in the following we will deal with it in the case where $t\in \{3,4\}$.

 \noindent Regarding the equivalence issue for $L_{h,3,1}$, this was completely solved in \cite{Bartoli_Zanella_Zullo}. As a consequence of Theorem \ref{main_result} Case $t=3$, we may state the following result
\begin{theorem}
If $h\in \mathbb{F}_{q^2}$, the linear set $L_{h,3,s}\subset \mathrm{PG}(1, q^6) $ is ${\PGaL}$-equivalent to some 
\[L_{\zeta}=\{\langle (x,x^{q}+x^{q^{3}}+\zeta x^{q^{5}})
\rangle_{\F_{q^6}} \colon x \in \F^*_{q^6}\}\]
where $\xi
 \in \F_{q^6}$ such that $\xi^2+\xi=1$ if and only if $h\in\mathbb{F}_q$ and $q$ is a power of 5. If $h\notin \mathbb{F}_{q^2}$, then $L_{h,3,s}$ is not $\PGaL$-equivalent to $L_{2,s}$, $L_{2,s,\eta}$ and $L_{6,s,\delta}$ in $\mathrm{PG}(1, q^6)$.
\end{theorem}

Now we investigate, the same issue for $L_{h,4,s}$. In \cite{longobardi_zanella}, by using projection techniques, the authors solved the problem for the linear set associated with the code $\mathcal{C}_t$, showing that this is not equivalent to any other linear sets previously constructed.\\
Following the same approach, we will show the same result for any value of the parameter $h$.
We will work in the following framework. Let $x_0, x_1, \dots, x_7$ be the homogeneous coordinates of $\mathrm{PG}(7,q^8)$ and consider $$\Sigma=\{\langle(x,x^q, \dots,x^{q^7})\rangle_{\mathbb{F}_{q^8}}:x\in \mathbb{F}_{q^8}^*\}$$ to be 
a fixed (canonical) subgeometry of $\mathrm{PG}(7,q^8)$. Let $$\hat{\sigma}: \langle( x_0, x_1, \dots, x_7)\rangle_{\mathbb{F}_{q^8}}\in \mathrm{PG}(7,q^8)\longrightarrow\langle(x^q_7, x^q_0, x^q_1, \dots, x^q_6)\rangle_{\mathbb{F}_{q^8}}\in \mathrm{PG}(7,q^8)$$
be a collineation of $\mathrm{PG}(7,q^8)$. It is straightforward to see that $\hat{\sigma}^m$ fixes $\Sigma$ pointwise for any $m\in\{0,\dots,7\},$ and $\hat{\sigma}^u$ is a generator of $\mathrm{Fix}(\Sigma)$ for any $u$ such that $(u,8)=1$.

As proven in \cite[Theorems 1 and 2]{lunardon_polverino} each linear set $L$  of $\PG(1,q^8)$ can be obtained by projecting $\Sigma$ from a $5$-dimensional projective subspace $\Gamma$ of $\PG(7,q^8)$ onto a line of $\PG(7,q^8)$ which is disjoint from $\Gamma$.  

By \cite[Theorem 1.1 and 3.5]{Zanella-Zullo}, as done in \cite[Result 3.4]{Bartoli_Zanella_Zullo}, we can assert that if $L$ is scattered and the least positive integer $\gamma$ such $$\dim(\Gamma \cap \Gamma^{\hat{\sigma}} \cap \cdots \cap \Gamma^{\hat{\sigma}^\gamma}) > 5-2\gamma,$$ does not belong to $\{1,2\}$, then $L$ is neither equivalent to $L_{2,s}$ nor to $L_{2,s,\eta}$ (for more details see also \cite{longobardi_zanella}). With this in mind, we prove the following 
\begin{theorem}\label{notLPtype}
The linear set $L_{h,4,s}$ is neither equivalent to $L_{2,s}$ nor to $L_{2,s,\eta}$.
\end{theorem}
\begin{proof}
 Let $h^{1+q^{4s}}=-1$ and $(s,4)=1$. The linear set $L_{h,4,s}$ can be seen as the projection of the canonical subgeometry $\Sigma$ from the  
\begin{equation}\label{system_1}
    \Gamma_s:\begin{cases}
    x_0=0\\
    x_s+x_{3s}-h^{1-q^{5s}}x_{5s}+h^{1-q^{7s}}x_{7s}=0
    \end{cases}
\end{equation}
onto the line $\mathrm{PG}(1,q^8) \subset \mathrm{PG}(7,q^8)$ with equations $x_{2s}= x_{3s}= \dots= x_{7s}=0$, where indices in system above are taken modulo $8$. 
Let \begin{equation}\label{system_2}
\Gamma_s^{\hat{\sigma}^u}:\begin{cases}
    x_{u}=0\\
    x_{s+u}+x_{3s+u}-h^{q^u-q^{5s+u}}x_{5s+u}+h^{q^u-q^{7s+u}}x_{7s+u}=0,
    \end{cases}
\end{equation}
where $u \in \{s,3s,5s,7s\} \pmod 8$. Then
\begin{equation}\label{intsystem}
    \Gamma_s \cap \Gamma_s^{\hat{\sigma}^u}:\begin{cases}
    x_0=0\\
    x_u=0\\
    x_s+x_{3s}-h^{1-q^{5s}}x_{5s}+h^{1-q^{7s}}x_{7s}=0\\
    x_{s+u}+x_{3s+u}-h^{q^u-q^{5s+u}}x_{5s+u}+h^{q^u-q^{7s+u}}x_{7s+u}=0.
    \end{cases}
\end{equation}
It is straightforward to see that, for any choice of $s$ and $u$, none of the unknowns in the third equation appears among those of the fourth one; hence, the four equations in System \eqref{intsystem} are independent. As a consequence, $\dim(\Gamma_s\cap\Gamma_s^{\hat{\sigma}^u})=3$.  Now consider, 
\begin{equation}\label{systemin3}
    \Gamma_s\cap\Gamma_s^{\hat{\sigma}^u}\cap\Gamma_s^{\hat{\sigma}^{2u}}:\begin{cases}
    x_0=0\\
    x_{u}=0\\
    x_{2u}=0\\
    x_s+x_{3s}-h^{1-q^{5s}}x_{5s}+h^{1-q^{7s}}x_{7s}=0\\
    x_{s+u}+x_{3s+u}-h^{q^u-q^{5s+u}}x_{5s+u}+h^{q^u-q^{7s+u}}x_{7s+u}=0\\
    x_{s+2u}+x_{3s+2u}-h^{q^{2u}-q^{5s+2u}}x_{5s+2u}+h^{q^{2u}-q^{7s+2u}}x_{7s+2u}=0.
    \end{cases}
\end{equation}
As before, we note that for any $u\in\{s,3s, 5s,7s\}$:
$$\{s+2u, 3s+2u, 5s+2u, 7s+2u\}=\{s,3s,5s,7s\} $$ and 
$$\{s+u,3s+u,5s+u, 7s+u\} = \{0,2s,4s,6s\},$$
where the integers are taken modulo $8$. So the fourth and sixth equations in System \eqref{systemin3} have the same unknowns; in addition none of the unknowns of the fifth equation is among these.
So, suppose that $u\equiv s \pmod 8$, then \eqref{systemin3} becomes
\begin{equation}\label{systemin3u=s}
    \Gamma_s\cap\Gamma_s^{\hat{\sigma}^u}\cap\Gamma_s^{\hat{\sigma}^{2u}}:\begin{cases}
    x_0=0\\
    x_{s}=0\\
    x_{2s}=0\\
    x_{3s}-h^{1-q^{5s}}x_{5s}+h^{1-q^{7s}}x_{7s}=0\\
    x_{4s}-h^{q^s-q^{6s}}x_{6s}=0\\
    x_{3s}+x_{5s}-h^{q^{2s}-q^{7s}}x_{7s}=0.
    \end{cases}
\end{equation}
By fourth and sixth equation, one gets 
$(1+h^{1-q^{5s}})x_{5s}=h^{1-q^{7s}}(1+h^{q^{2s}-1})x_{7s}$ and since $h^{q^{2s}-1} \not = -1$ ( see \cite[Proposition 3.2]{longobardi_marino_trombetti_zhou}), all equations in \eqref{systemin3u=s} are independent.
The same argument can be applied to all the other congruences leading to $\dim(\Gamma_s\cap\Gamma_s^{\hat{\sigma}^u}\cap\Gamma_s^{\hat{\sigma}^{2u}})=1$ for any $u \in \{s,3s,5s,7s\} \pmod 8$.
Hence, we have that $\gamma\notin\{1,2\}$; therefore, $L_{h,4,s}$ is neither equivalent to $L_{2,s}$, nor $L_{2,s,\eta}$. 
\end{proof}

As a consequence of this achievement we may derive the following result about equivalence of the codes $\mathcal{C}_{h,t,s}$ with ones belonging to other families.

\begin{theorem}\label{th:inequivalence}
Let  $n=2t,$ $t \in \{3,4\}$, $h, k,\eta \in \F_{q^{n}}$ satisfying $\N_{q^{n}/q^t}(h)=  \N_{q^n/q^t}(k)= -1$ and $\N_{q^n/q}(\eta) \not = 1$. Let $s \in \mathbb{N}$ such that $(n,s)=1$.  

\begin{itemize}
\item[(a)] ${\mathcal H}_{2,s}(\eta)$ and \,$\mathcal{C}_{h,t,s}$ are not equivalent. 
\item[(b)] Assume $\delta \in \F_{q^{2t}}$ such that $\N_{q^n/q^{n/2}}(\delta) \not \in \{0,1\},$ and the other conditions on $\delta$ and $q$ as expressed in \cite[Section 7]{Csajbok_Marino_Polverino_Zanella}, hold true. Then, the codes ${\cal K}_{2t,s,\delta}$ and $\mathcal{C}_{h,t,s}$ are not equivalent.
\item [(c)] Assume $\zeta \in \F_{q^{6}}$ such that $\zeta^2+\zeta=1$. Then, the codes $\mathcal{Z}_{6,\zeta}$ and $\cC_{h,3,s}$  are not equivalent except $h \in \F_{q}$ and $q$ a power of $5$.
\end{itemize}

\end{theorem} 

\begin{proof}  
$(a)$ First assume $\eta=0$, then $\mathcal{H}_{2,s}(0)=\mathcal{G}_{2,s}$. This has left and right idealizers both isomorphic to $\F_{q^n}$.  
Since $I_{R}(\cC_{h,t,s})$ is isomorphic to $\F_{q^2}$ for both $t=3$ and $4$, then $\cC_{h,t,s}$ can not be equivalent to $\mathcal{G}_{2,s}$. On the other hand, if $\eta \neq 0$ and $t=3$, by\, \cite[Result 3.4]{Bartoli_Zanella_Zullo}, we get that $\cC_{h,t,s}$ is not equivalent to ${\mathcal H}_{2,s}(\eta)$. For $t=4$ the result directly follows from Theorem \ref{notLPtype}.\\
\noindent$(b)$\,  The statement is a direct consequence of the fact that  the right idealizer of ${\cal K}_{2t,s,\delta}$ is isomorphic to $\F_{q^t}$, \cite{Csajbok_Marino_Polverino_Zanella}.\\
$(c)$ The statement is a direct consequence of Theorem \ref{main_result} and \cite[Proposition 3.10]{Bartoli_Zanella_Zullo}.
\end{proof}

In \cite{Csajbok_Marino_Polverino_Zanella}, the  scattered subspace
\begin{equation}
    U_{n,s,\delta}=\{(x,\delta x^{q^{s}}+x^{q^{s+ n/2}}) \colon x \in \F_{q^n}\}
\end{equation}
associated with the code $\mathcal{K}_{n,s,\delta},$
was studied. As a direct consequence of \cite[Propositions 4.1 and 4.2]{Csabok_Marino_Zullo}, we can say that the scattered linear set associated with $U_{8,s,\delta}$, is equivalent to $L_{h,4,s}$ if and only if the underlying subspaces are $\GaL$-equivalent and, hence, if and only if the MRD codes $\cC_{h,4,s}$ and $\mathcal{K}_{8,s,\delta}$ are equivalent. However, as shown in Theorem \ref{th:inequivalence}$(b)$, this is not the case. Therefore, we can state the following

\begin{theorem}
The linear set $L_{h,4,s}$ is not
$\PGaL$-equivalent to $L_{2,s}$ and $L_{2,s,\eta}$  and $L_{8,s,\delta}$ of $\PG(1,q^8)$.
\end{theorem}

\section{Acknowledgement}
We wish to thank anonymous referees for their valuable comments and suggestions which greatly improved the readability of the article. We are also grateful to Giuseppe Marino for helpful discussions and technical advices.

\noindent
Somi Gupta, Giovanni Longobardi and Rocco Trombetti\\
Dipartimento di Matematica e Applicazioni
 “Renato Caccioppoli”\\
Università degli Studi di Napoli Federico II\\
via Cintia, Monte S. Angelo
I\\
80126 Napoli - Italy\\
email:
\texttt{\{somi.gupta, 
giovanni.longobardi, rtrombet\}@unina.it}

\end{document}